\documentclass{scrartcl}

\usepackage[dvipsnames]{xcolor}
\usepackage{hyperref}
\usepackage{todonotes}
\usepackage{microtype}
\usepackage{mathpartir}
\usepackage{pftools}
\usepackage[a4paper]{geometry}
\usepackage{fullpage}
\usepackage{amsmath}
\usepackage{amsthm}
\usepackage{stmaryrd}
\usepackage{esvect}
\usepackage{multicol}
\usepackage{natbib}
\usepackage{xfrac}
\usepackage{lstcoq}

\usepackage{includes}

\newtheorem{definition}{Definition}[section]
\newtheorem{theorem}[definition]{Theorem}
\newtheorem{lemma}[definition]{Lemma}
\newtheorem{example}[definition]{Example}
\newtheorem{remark}[definition]{Remark}
\newtheorem{corollary}[definition]{Corollary}

\hypersetup{colorlinks=true, citecolor=blue!70!black, linkcolor=blue!70!black, urlcolor=blue!70!black}

\title {Consistency of the Predicative Calculus of Cumulative Inductive Constructions (\pCuIC{})}

\author{
Amin Timany
\footnote{This research was partly carried out
  while I was visiting Inria Paris and Universit\'e Paris Diderot and
  partly while I was visiting Aarhus university.}
  \\ imec-Distrinet, KU Leuven, Belgium \\ amin.timany@cs.kuleuven.be
 \and
 Matthieu Sozeau \\ Inria Paris \& IRIF, France \\ matthieu.sozeau@inria.fr}

\date{May 11, 2018}

\begin{document}
\maketitle

\begin{abstract}
  In order to avoid well-know paradoxes associated with
  self-referential definitions, higher-order dependent type theories
  stratify the theory using a countably infinite hierarchy of
  universes (also known as sorts), $\Type{0} : \Type{1} :
  \cdots$. Such type systems are called cumulative if for any type
  $\Typ$ we have that $\Typ : \Type{i}$ implies $\Typ :
  \Type{i+1}$. The predicative calculus of inductive constructions
  (\pCIC{}) which forms the basis of the \Coq{} proof assistant, is
  one such system.

  In this paper we present and establish the soundness of the
  predicative calculus of cumulative inductive constructions
  (\pCuIC{}) which extends the cumulativity relation to inductive
  types.
\end{abstract}

\section{Introduction}
We construct a model for \pCuIC{} based inspired by the model of
\citet{DBLP:journals/corr/abs-1111-0123} establishing consistency of
\pCuIC{}.  In loc. cit. authors present a model of \pCIC{} without
inductive types in the sort $\Prop$. Similarly, the model that we
shall construct also does not feature inductive types in the sort
$\Prop$. We note, however, that most inductive types in $\Prop$ can be
encoded using their Church encoding. For instance, the type
\Coqe{False} and conjunction of two predicates can be defined as
follows:
\begin{coq}
Definition conj (P Q : Prop) :=
  \forall (R : Prop), (P -> Q -> R) -> R.

Definition False := \forall (P : Prop), P.
\end{coq}

The model of \citet{DBLP:journals/corr/abs-1111-0123} does not support
the sort $\Set$, neither the predicative version nor the impredicative
one. In this paper, we treat the predicative sort $\Set$ simply as a
shorthand for $\Type{0}$. Note that in \Coq{} \emph{it does not hold
  that} $\Prop : \Set$. However, in our model the interpretation of
$\Prop$ is an element of the interpretation of $\Type{0}$ but this
does not preclude our model from being sound for the typing rules of
the system. In addition, the model of
\citet{DBLP:journals/corr/abs-1111-0123} does support $\eta$-expansion
yet this is not included in their rules for judgemental equality. We
do include $\eta$-expansion in our system.

The main difference between our system and the one by
\citet{DBLP:journals/corr/abs-1111-0123}, apart from cumulativity for
inductive types which is a new contribution to \pCIC{}, is that in
\citet{DBLP:journals/corr/abs-1111-0123} the system features case
analysis and fixpoints (recursive functions on inductive types) for
inductive types while ours features eliminators. It is well-known that
all functions written using fixpoints and case analysis on inductive
types are representable using eliminators. However, the former is
closer to actual implementation in \Coq{}. On the other hand, the
occurrences of recursive calls in recursive functions are may be
hidden under some computations. That is, in order to obtain the
position and arguments of recursive calls, which are crucial for
constructing the model of recursive functions, in general requires
normalizing the body of the recursive function. It so appears that the
intention of \citet{DBLP:journals/corr/abs-1111-0123} is to construct
a set theoretic model for \Coq{}, assuming the normalization property
which already implies soundness of their system. This, however, is not
possible when constructing a model that is intended to establish the
soundness of the system. Indeed, our model, similarly to the model
of \citet{DBLP:journals/corr/abs-1111-0123}, does not imply
normalization of well-typed terms. Therefore, when modeling recursive
functions, i.e., fixpoints in \Coq{}, we cannot assume that we can
normalize the body of the fixpoint. This is the main reason why we
opted for formalizing our system with eliminators instead.

\subsection{Contributions}
\citet{DBLP:conf/ictac/Timany015} give an account of then
work-in-progress on extending \pCIC{} with a single cumulativity rule
for cumulativity of inductive types. The authors show a rather
restricted subsystem of the system that they present to be sound. This
subsystem roughly corresponds to the fragment where terms of
cumulative inductive types do not appear as dependent arguments in
other terms.  The proof given in \citet{DBLP:conf/ictac/Timany015} is
done by giving a syntactic translation from that subsystem to
\pCIC{}. In this paper, we extend and complete the work that was
initiated by \citet{DBLP:conf/ictac/Timany015}.

In particular, in this work, we consider a more general version of the
cumulativity rule for inductive types. Adding to this, we also
consider related rules for judgemental equality of inductive types
which are given rise to by the mutual cumulativity relation and also
judgemental equality of the terms constructors of types in the
cumulativity relation.

In this work, we present the system \pCuIC{} in its full details and
prove that the system as presented to be sound. We do this by
constructing a set-theoretic model in \ZFC{}, together with the axiom
that there are countably many uncountable strongly inaccessible
cardinals, inspired by the model of
\citet{DBLP:journals/corr/abs-1111-0123}. The cumulativity of
inductive types as presented in this paper is now supported in
\Coq{}\,8.7 \citep{the_coq_development_team_2017_1003421}.

\paragraph{Notations for different equalities and the like:}
\begin{center}
\begin{tabular}{|c|l|}
\hline
$a \eqdef b$ & $a$ is defined as $b$\\
\hline
$a = b$ & equality of mathematical objects, e.g., sets, sequences, etc.\\
\hline
$\term \synteq \term'$ & Syntactically identical terms \\
\hline
$\jueq{\ctx}{\term}{\term'}{\Typ}$ &
\begin{tabular}{@{}l@{}}
  Judgemental equality, i.e., $\term$ and $\term'$ are judgementally equal\\
  terms of type $\Typ$ under context $\ctx$
\end{tabular} \\
\hline
\end{tabular}
\end{center}

\section{\pCuIC{}}\label{sect:pcuic}
The terms and contexts of the language \pCuIC{} is as follows:
\begin{align*}
\var, \varB, \varC, \dots & & (\textrm{Variables})\\
\sort ::={} & \Prop, \Type{0}, \Type{1}, \dots & (\textrm{Sorts})\\
\term, \termB, \dots, \Term, \TermB, \dots, & \\
  \Typ, \TypB, \dots ::={} & \var \ALT \sort \ALT
  \Forall \var : \Typ. \TypB \ALT \Lam \var : \Typ. \term \ALT \\
  & \Let \var := \term : \Typ in \termB \ALT \Term~\TermB \ALT & (\textrm{Terms})\\
  & \Inds.\var \ALT \Elim{\term}{\Inds.\indtyp_i; \vv{\termB}}{\motive_{\indtyp_1}, \dots, \motive_{\indtyp_n}}{\genfun_{\constr_1},\dots, \genfun_{\constr_m}}\\
\Inds ::={} & \Ind_{n}\{\declinds := \declcons\} & (\textrm{Inductive blocks})\\
\decl ::={} & \cdot \ALT \decl, \var : \Typ & (\textrm{Declarations}) \\
\ctx ::={} & \cdot \ALT \ctx, \var : \Typ, \ALT \ctx, \var := \term : \Typ \ALT
\ctx, \Inds & (\textrm{Contexts})
\end{align*}
Note that although by abuse of notation we write
$\var : \Typ \in \ctx$, $\var := \term : \Typ \in \ctx$ or
$\var : \Typ \in \decl$, contexts and declarations \emph{are not}
sets. In particular, the order in declarations is important, this is
even the case for declarations where there can be no dependence among
the elements. We write $\decl(\var)$ to refer to $\Typ$ whenever
$\var : \Typ \in \decl$.

Here, $\Ind_{n}\{\declinds := \declcons\}$ is a block of mutual
inductive definitions where $n$ is the number of parameters, the
declarations in $\declinds$ are the inductive types of the block and
declarations in $\declcons$ are constructors of the block.  The term
$\Ind_{n}\{\declinds := \declcons\}.\var$ is an inductive definition
whenever $\var \in \declinds$ and a constructor whenever
$\var \in \declcons$. The term
$\Elim{\term}{\Inds.d_k; \vv{\termB}}{Q_{d_1}, \dots, Q_{d_n}}{f_{c_1},\dots,
  f_{c_m}}$ is the elimination of $\term$ (a term of the inductive
type $\Inds.{d_k}$ applied to parameters $\vv{\termB}$), $Q_{d_i}$'s are the motives of elimination, i.e.,
the result of the elimination will have type
$Q_{d_k}~\vv{\termB}~\vv{\varterm}~\term$ whenever the term being eliminated,
$\term$, has type $\Inds.{d_k}~\vv{\termB}~\vv{\varterm}$. The term $f_{c_i}$ in
the eliminator above is a \emph{case-eliminator} corresponding to the
case where $t$ is constructed using constructor
$\Inds.c_i$.

We write $\len(\vv{s})$, $\len(p)$, $\dom(\genfun)$, $\dom(\decl)$ and
$\dom(\ctx)$ respectively for the length of a sequence $\vv{s}$,
length of a tuple $p$, the domain of a partial map $\genfun$, domain
of a declaration $\decl$ or domain of a context $\ctx$.  Notice
crucially that the inductive types and constructors of an inductive
block are \emph{not} part of the domain of the context that they
appear in.  We write $\nil$ for the empty sequence. We write
$\IndsOf(\ctx)$ for the sequence of inductive types in the context
$\ctx$. We write tuples as $\tuple{a_1; a_2; \dots; a_n}$.

\begin{definition}[Free variables]
\begin{align*}
  \intertext{Free variables of terms}
  \freevars(\Prop) \eqdef{}
  & \emptyset\\
  \freevars(\Type{i}) \eqdef{}
  & \emptyset\\
  \freevars(\varC) \eqdef{}
  & \set{\varC}\\
  \freevars(\Lam \varB : \Typ. \termB) \eqdef{}
  & \freevars(\Typ) \cup (\freevars(\termB) \setminus \set{\varB})\\
  \freevars(\Let \varB := \termB : \Typ in \termC) \eqdef{}
  & \freevars(\Typ) \cup \freevars(\termB) \cup (\freevars(\termC) \setminus \set{\varB})\\
  \freevars(\termB~\termC) \eqdef{}
  & \freevars(\termB) \cup \freevars(\termC)\\
  \freevars(\Inds.\varC) \eqdef{}
  & \freevars(\Inds)\\
  \freevars(\Elim{\termB}{\Inds.d_i;\vv{\termB}}{\vv{\motive}}{\vv{\genfun}}) \eqdef{}
  & \freevars(\termB) \cup \freevars(\Inds) \cup \freevars(\vv{\termB}) \cup \freevars(\vv{\motive}) \cup \freevars(\vv{\genfun})\\
  \intertext{Free variables of inductive blocks}
  \freevars(\Ind_{n}\{\declinds := \declcons\}) \eqdef{}
  & \freevars(\declinds) \cup \freevars(\declcons)\\
  \intertext{Free variables of sequences of terms}
  \freevars(\nil) \eqdef{}
  & \nil\\
  \freevars(\termC, \vv{\termC}) \eqdef{}
  & \freevars(\termC) \cup \freevars(\vv{\termC})\\
  \intertext{Free variables of declarations}
  \freevars(\cdot) \eqdef{}
  & \emptyset\\
  \freevars(\varB : \Typ, \decl) \eqdef{}
  & \freevars(\Typ) \cup \freevars(\decl)
\end{align*}
\end{definition}

We define simultaneous substitution for terms as follows:
\begin{definition}[Simultaneous substitution]
  We assume that $\vv{\var}$ is a sequence of distinct variables. In
  this definition for the sake of simplicity we use $\varB$ for all
  bound variables.
\begin{align*}
  \intertext{Substitution for terms}
  \subst{\Prop}{\vv{\var}}{\vv{\term}} \eqdef{}
  & \Prop\\
  \subst{\Type{i}}{\vv{\var}}{\vv{\term}} \eqdef{}
  & \Type{i}\\
  \subst{\varC}{\vv{\var}}{\vv{\term}} \eqdef{}
  & \term_i \text{ if } \var_i = \varC\\
  \subst{\varC}{\vv{\var}}{\vv{\term}} \eqdef{}
  & \varC \text{ if } \forall i.~\var_i \neq \varC\\
  \subst{(\Lam \varB : \Typ. \termB)}{\vv{\var}}{\vv{\term}} \eqdef{}
  & \Lam \varB : \subst{\Typ}{\vv{\var}}{\vv{\term}}.
    \subst{\termB}{\vv{\var}'}{\vv{\term}'}\\
  \subst{(\Let \varB := \termB : \Typ in \termC)}{\vv{\var}}{\vv{\term}} \eqdef{}
  & \Let \varB := \subst{\termB}{\vv{\var}}{\vv{\term}} : \subst{\Typ}{\vv{\var}}{\vv{\term}}
    in \subst{\termC}{\vv{\var}'}{\vv{\term}'}\\
  \subst{(\termB~\termC)}{\vv{\var}}{\vv{\term}} \eqdef{}
  & \subst{\termB}{\vv{\var}}{\vv{\term}}~\subst{\termC}{\vv{\var}}{\vv{\term}}\\
  \subst{\Ind_{n}\{\declinds := \declcons\}.\varC}{\vv{\var}}{\vv{\term}} \eqdef{}
  & \Ind_{n}\{\subst{\declinds}{\vv{\var}}{\vv{\term}} := \subst{\declcons}{\vv{\var}}{\vv{\term}}\}.\varC\\
  \subst{\Elim{\termB}{\Inds.d_i;\vv{\termB}}{\vv{\motive}}{\vv{\genfun}}}{\vv{\var}}{\vv{\term}} \eqdef{}
  & \Elim{\subst{\termB}{\vv{\var}}{\vv{\term}}}{\subst{\Inds.d_i}{\vv{\var}}{\vv{\term}}; \subst{\vv{\termB}}{\vv{\var}}{\vv{\term}}}{\subst{\vv{\motive}}{\vv{\var}}{\vv{\term}}}{\subst{\vv{\genfun}}{\vv{\var}}{\vv{\term}}}\\
  \intertext{Substitution for inductive blocks}
  \subst{\Ind_{n}\{\declinds := \declcons\}}{\vv{\var}}{\vv{\term}} \eqdef{}
  & \Ind_{n}\{\subst{\declinds}{\vv{\var}}{\vv{\term}} := \subst{\declcons}{\vv{\var}}{\vv{\term}}\}\\
  \intertext{Substitution for sequences}
  \subst{\nil}{\vv{\var}}{\vv{\term}} \eqdef{}
  & \nil\\
  \subst{\termC, \vv{\termC}}{\vv{\var}}{\vv{\term}} \eqdef{}
  & \subst{\termC}{\vv{\var}}{\vv{\term}}, \subst{\vv{\termC}}{\vv{\var}}{\vv{\term}}\\
  \intertext{Substitution for declarations}
  \subst{\cdot}{\vv{\var}}{\vv{\term}} \eqdef{}
  & \cdot\\
  \subst{\varB : \Typ, \decl}{\vv{\var}}{\vv{\term}} \eqdef{}
  & \varB : \subst{\Typ}{\vv{\var}}{\vv{\term}}, \subst{\decl}{\vv{\var}}{\vv{\term}}\\
  \intertext{Substitution for context}
  \subst{\cdot}{\vv{\var}}{\vv{\term}} \eqdef{}
  & \cdot\\
  \subst{\varB : \Typ, \decl}{\vv{\var}}{\vv{\term}} \eqdef{}
  & \varB : \subst{\Typ}{\vv{\var}}{\vv{\term}}, \subst{\decl}{\vv{\var'}}{\vv{\term'}}\\
  \intertext{where $\vv{\var'} = \vv{\var}$ and $\vv{\term'} = \vv{\term}$ if $\varB$ does not appear in $\vv{\var}$ and is as follows whenever $\var_i = \varB$:
  \[ \vv{\var'} = \var_0, \dots, \var_{i-1}, \var_{i+1}, \dots, \var_n \]
  \[ \vv{\term'} = \term_0, \dots, \term_{i-1}, \term_{i+1}, \dots, \term_n \]
  }
\end{align*}
\end{definition}

\subsection{Basic constructions}
Figure~\ref{fig:pCuIC:basic-cons} shows typing rules for the basic
constructions, i.e., well-formedness of contexts ($\WFctx{\ctx}$),
sorts, let bindings, dependent products (also referred to as dependent
functions), lambda abstractions and applications. It also contains the
rules for the judgemental equality for these constructions.  In this
figure, the relation $\ProdRs$ indicates the sort that a (dependent)
product type belongs to. The sort $\Prop$ is \emph{impredicative} and
therefore any product type with codomain in $\Prop$ also belongs to
$\Prop$.

\begin{figure}
\begin{mathpar}
\inferH{WF-ctx-empty}{}{\WFctx{\cdot}}
\and
\inferH{WF-ctx-hyp}{\typed{\ctx}{\Typ}{\sort} \and
  \var \not\in \dom(\ctx)}{\WFctx{\ctx, \var : \Typ}}
\and
\inferH{WF-ctx-def}{\typed{\ctx}{\term}{\Typ} \and \var \not\in \dom(\ctx)}
   {\WFctx{\ctx, (\var := \term : \Typ)}}
\and
\inferH{Prop}{\WFctx{\ctx}}{\typed{\ctx}{\Prop}{\Type{i}}}
\and
\inferH{Hierarchy}{\WFctx{\ctx} \and i < j}{\typed{\ctx}{\Type{i}}{\Type{j}}}
\and
\inferH{Var}{\WFctx{\ctx} \and \var : \Typ \in \ctx \and \textrm{or} \and (\var := \term : \Typ) \in \ctx}{\typed{\ctx}{\var}{\Typ}}
\and
\inferH{Let}{\typed{\ctx, (\var := \term : \Typ)}{\termB}{\TypB}}
  {\typed{\ctx}{\Let \var := \term : \Typ in \termB}{\subst{\TypB}{\var}{\term}}}
\and
\inferH{Let-eq}{\jueq{\ctx}{\Typ}{\Typ'}{\sort} \and \jueq{\ctx}{\term}{\term'}{\Typ} \and
    \jueq{\ctx, (\var := \term : \Typ)}{\termB}{\termB'}{\TypB}}
  {\jueq{\ctx}{\Let \var := \term : \Typ in \termB}{\Let \var := \term' : \Typ' in \termB'}
     {\subst{\TypB}{\var}{\term}}}
\and
\inferH{Prod}{\typed{\ctx}{\Typ}{\sort_1} \and
    \typed{\ctx, \var : \Typ}{\TypB}{\sort_2} \and \ProdRs(\sort_1, \sort_2, \sort_3)}
  {\typed{\ctx}{\Forall \var : \Typ. \TypB}{\sort_3}}
\and
\inferH{Prod-eq}{\jueq{\ctx}{\Typ}{\Typ'}{\sort_1} \and
    \jueq{\ctx, \var : \Typ}{\TypB}{\TypB'}{\sort_2} \and \ProdRs(\sort_1, \sort_2, \sort_3)}
  {\jueq{\ctx}{\Forall \var : \Typ. \TypB} {\Forall \var : \Typ'. \TypB'}{\sort_3}}
\and
\inferH{Lam}{\typed{\ctx, \var : \Typ}{\Term}{\TypB} \and
    \typed{\ctx}{\Forall \var : \Typ. \TypB}{\sort}}
  {\typed{\ctx}{\Lam \var : \Typ. \Term}{\Forall \var : \Typ. \TypB}}
\and
\inferH{Lam-eq}{\jueq{\ctx}{\Typ}{\Typ'}{\sort_1} \and
    \jueq{\ctx, \var : \Typ}{\Term}{\Term'}{\TypB} \and
    \typed{\ctx}{\Forall \var : \Typ. \TypB}{\sort_2}}
  {\jueq{\ctx}{\Lam \var : \Typ. \Term}{\Lam \var : \Typ'. \Term'}
     {\Forall \var : \Typ. \TypB}}
\and
\inferH{App}{\typed{\ctx}{\Term}{\Forall \var : \Typ. \TypB} \and
    \typed{\ctx}{\TermB}{\Typ}}
  {\typed{\ctx}{\Term~\TermB}{\subst{\TypB}{\var}{\TermB}}}
\and
\inferH{App-eq}{\jueq{\ctx}{\Term}{\Term'}{\Forall \var : \Typ. \TypB} \and
    \jueq{\ctx}{\TermB}{\TermB'}{\Typ}}
  {\jueq{\ctx}{\Term~\TermB}{\Term'~\TermB'}
    {\subst{\TypB}{\var}{\TermB}}}
\and
\inferH{Predicativity}{}{\ProdRs(\Type{i}, \Type{j}, \Type{\Max{i, j}})}
\and
\inferH{Predicativity'}{}{\ProdRs(\Prop, \Type{i}, \Type{i})}
\and
\inferH{Impredicativity}{}{\ProdRs(s, \Prop, \Prop)}
\end{mathpar}
\caption{Basic construction}
\label{fig:pCuIC:basic-cons}
\end{figure}

\subsection{Inductive types and their eliminators}
The typing rules for inductive types, their constructors and
eliminators are depicted in
Figure~\ref{fig:pCuIC:inductives-eliminators}.

\begin{figure}
\begin{mathpar}
\inferH{Ind-WF}{\indwf n \ctx \declinds \declcons \and 
    (A \synteq \Forall \termP : \TypP.\Forall \vartermC : \Term.\arity_\indtyp \and \typed{\ctx}{\Typ}{\sort_\indtyp} \text{ for all } (\indtyp : A) \in \declinds) \and \\\
    \varTyp \synteq \Forall \termP : \TypP.\varTyp' \and \typed{\ctx, \declinds, \termP : \TypP}{\varTyp'}{\arity_\indtyp} \text{ for all } (\constr : \varTyp) \in \declcons \text{ if }
\constr \in \constrsof(\declcons, \indtyp)
}{\WFctx{\ctx, \Ind_n\{\declinds := \declcons\}}}
\and
\inferH{Ind-type}{\WFctx{\ctx} \and \Inds \synteq \Ind_n\{\declinds := \declcons\} \in \ctx
    \and \indtyp_i \in \dom(\declinds)}
  {\typed{\ctx}{\Inds.d_i}{\declinds(\indtyp_i)}}
\and
\inferH{Ind-constr}{\WFctx{\ctx} \and \Inds \synteq \Ind_n\{\declinds := \declcons\} \in \ctx
    \and \constr \in \dom(\declcons)}
  {\typed{\ctx}{\Inds.\constr}{\subst{\declcons(\constr)}{\vv{\indtyp}}{\vv{\declinds.\indtyp}}}}
\and
\inferH{Ind-Elim}{\WFctx{\ctx} \and \Inds \synteq \Ind_n\{\declinds := \declcons\} \in \ctx
  \and \\\
 \dom(\declinds) = \set{\indtyp_1, \dots, \indtyp_l} \and \dom(\declcons) = \set{\constr_1, \dots, \constr_{l'}} \and \\\
\typed{\ctx}{\motive_{\indtyp_i}}{\Forall \vv{\var}:\vv{\Typ}. (\indtyp_i~\vv{\var}) \to \sort'} \text{ where }
\declinds(\indtyp_i) \synteq \Forall \vv{\var}:\vv{\Typ}.\sort \text{ for all } 1 \leq i \leq l \and \\\
\typed{\ctx}{\term}{\Inds.\indtyp_k~\vv{\termB}~\vv{\vartermC}} \and \len(\vv{\termB}) = n \and
\typed{\ctx}{\genfun_{\constr_i}}
{\elimtyp{\Inds}{\vv{\motive}}(\constr_i, \declcons(\constr_i))} \text{ for all } 1 \le i \le l'
}
{\typed{\ctx}{\Elim{\term}{\Inds.\indtyp_k; \vv{\termB}}{\motive_{\indtyp_1}, \dots, \motive_{\indtyp_l}}
{\genfun_{\constr_1},\dots,\genfun_{\constr_{l'}}}}{\motive_{\indtyp_k}~\vv{\termB}~\vv{\vartermC}~\term}}
\and
\inferH{Ind-Elim-eq}{\WFctx{\ctx} \and \Inds \synteq \Ind_n\{\declinds := \declcons\} \in \ctx
  \and \\\
 \dom(\declinds) = \set{\indtyp_1, \dots, \indtyp_l} \and \dom(\declcons) = \set{\constr_1, \dots, \constr_{l'}} \and \\\
(\jueq{\ctx}{\motive_{\indtyp_i}}{\motive_{\indtyp_i}'}{\Forall \vv{\termP} : \vv{\TypP}. \Forall \vv{\var}:\vv{\Typ}. (\indtyp_i~\vv{\var}) \to \sort'} \text{ where } \\\
\declinds(\indtyp_i) \synteq \Forall \vv{\termP} : \vv{\TypP}. \Forall \vv{\var}:\vv{\Typ}.\sort \text{ for all } 1 \leq i \leq l) \and \\\
\jueq{\ctx}{\term}{\term'}{\Inds.\indtyp_k~\vv{\termB}~\vv{\vartermC}} \and \len(\vv{\termB}) = n \and
\jueq{\ctx}{\vv{\termB}}{\vv{\termB'}}{\vv{\TypP}} \and \\\
\jueq{\ctx}{\genfun_{\constr_i}}{\genfun_{\constr_i}'}
{\elimtyp{\Inds}{\vv{\motive}}(\constr_i, \declcons(\constr_i))} \text{ for all } 1 \le i \le l'
}
{
\ctx \vdash {\Elim{\term}{\Inds.\indtyp_k; \vv{\termB}}{\motive_{\indtyp_1}, \dots, \motive_{\indtyp_l}}
{\genfun_{\constr_1},\dots,\genfun_{\constr_{l'}}}} \simeq \\\
{\Elim{\term'}{\Inds.\indtyp_k; \vv{\termB'}}{\motive_{\indtyp_1}', \dots, \motive_{\indtyp_l}'}
{\genfun_{\constr_1}',\dots,\genfun_{\constr_{l'}}'}} : {\motive_{\indtyp_k}~\vv{\termB}~\vv{\vartermC}~\term}
}
\end{mathpar}
\caption{Inductive types and eliminators}
\label{fig:pCuIC:inductives-eliminators}
\end{figure}

\paragraph{Well-formedness of inductive types}
The first rule in this figure is the well-formedness of inductive
types. It states that in order to have that the context $\ctx$ is
well-formed after adding the mutual inductive block
$\Ind_n\{\declinds := \declcons\}$, i.e.,
$\WFctx{\ctx, \Ind_n\{\declinds := \declcons\}}$, all inductive types
in the block as well as all constructors need to be well-typed terms.
In addition, we must have that the well-formedness side condition
$\indwf{n}{\ctx}{\declinds}{\declcons}$ holds. The well-formedness
side condition $\indwf{n}{\ctx}{\declinds}{\declcons}$ holds whenever:
\begin{itemize}
\item All variables in $\declinds$ and $\declcons$ are distinct.
\item The first $n$ arguments of all inductive types and constructors
  in the block are the parameters. In other words, there is a sequence
  of terms $\vv{\TypP}$ such that $\len(\vv{\TypP}) = n$ and for all
  $\var : \varTyp \in \declinds, \declcons$ we have
  $\varTyp \synteq \Forall \vv{\termP} : \vv{\TypP}. \varTypB$ for
  some $\varTypB$.
\item Parameters are parametric. In other words, for all
  $\constr : \varTyp \in \declcons$ we have
  $\varTyp \synteq \Forall \vv{\termP} : \vv{\TypP}. \Forall \vv{\var}
  : \vv{\TypB}. \indtyp~\vv{\termP}~\vv{\termC}$. That is, every
  constructor constructs a term of some inductive type in the block
  where values applied for parameter arguments of the inductive type
  are parameter arguments of the constructor.
\item Every inductive type is just a type (an element of a universe)
  that depends on a number of arguments beginning with parameters. The
  non-parameter arguments are called the arity of the inductive
  type. In other words, for all $\indtyp \in \dom(\declinds)$ we have
  $\declinds(\indtyp) \synteq \Forall \vv{\var} : \vv{\TypP}. \Forall
  \vv{\vartermC} : \vv{\Term}. \arity_\indtyp$ where $\vv{\Term}$ are
  called the indices of the inductive type and $\arity_\indtyp$ is a
  sort called the arity of the inductive type $\indtyp$. We require
  that $\arity_\indtyp \neq \Prop$.
\item Every constructor is a constructs terms of an inductive type in
  the block. In other words, for all $\constr \in \dom(\declcons)$ we
  have $\constr \in \constrsof(\declcons, \indtyp)$ for some
  $\indtyp \in \declinds$ where $\constrsof(\declcons, \indtyp)$ is
  the set of constructors in $\declcons$ that construct terms of the
  inductive type $\indtyp$.
\[\constrsof(\declcons, \indtyp) \eqdef \set{\constr \in \dom(\declcons)
    \mid \declcons(\constr) \synteq \Forall \vv{\termP} : \vv{\TypP}.
    \Forall \vv{\var} : \vv{\varTypB}.d~\vv{\termB} }\]
\item All inductive types in the block appear only strictly positively
  in constructors. In other words, for all
  $\constr \in \dom(\declcons)$,
  $\strictpos{\dom(\declinds)}(\declcons(\constr))$ where
  $\strictpos{S}(\term)$ is determined by the following rules:
  \begin{mathpar}
    \infer{\termB \text{ does not appear in } \vv{\Typ} \text{ or } \vv{\term} \text{ for all } \termB \in S \and
      \TypB \in S}
    {\strictposarg{S}(\Forall \vv{\var} : \vv{\Typ}. \TypB~\vv{\term})}
    \and
    \infer{\termB \text{ does not appear in } \vv{\term} \text{ for all } \termB \in S \and
      \TypB \in S}{\strictpos{S}(\TypB~\vv{\term})}
    \and
    \infer{\strictposarg{S}(\Typ) \and \strictpos{S}(\TypB)}{\strictpos{S}(\Typ \to \TypB)}
    \and
    \infer{\termB \text{ does not appear in } \vv{\Typ} \text{ for all } \termB \in S \and \strictpos{S}(\TypB)}
    {\strictpos{S}(\Forall \vv{\var} : \vv{\Typ}. \TypB)}
  \end{mathpar}
\end{itemize}

\paragraph{Inductive types and constructors}
Rules \textsc{Ind-Type} and \textsc{Ind-constr} indicate,
respectively, the type of inductive types and constructs in a block of
mutually inductive types. The type of an inductive type of a block is
exactly the same as declared in the block. The type of constructors of
a block is determined by the type declared in the block except that
inductive types in block are replaced by their proper (global) names.

\begin{example}\label{ex:upolylist}
  The following is the definition natural numbers in\pCuIC{}.
  \begin{align*}
    \mathcal{N} \eqdef{} & \mathbfsf{Ind}_0\{\mathit{nat} : \Type{0} :=
    \mathit{zero} : \mathit{nat}, \mathit{succ} : \mathit{nat} \to \mathit{nat}\}\\
    \mathit{nat} \eqdef{} & \mathcal{N}.\mathit{nat}\\
    \mathit{zero} \eqdef{} & \mathcal{N}.\mathit{zero}\\
    \mathit{succ} \eqdef{} & \mathcal{N}.\mathit{succ}
  \end{align*}
\end{example}

\begin{example} The following is the definition of universe polymorphic
  lists (for level $i$) in \pCuIC{}.
  \begin{align*}
    \mathcal{L}_{i} \eqdef{} & \mathbfsf{Ind}_1\{\mathit{list} : \Forall \Typ : \Type{i}. \Type{i} :=
    \mathit{nil} : \Forall \Typ : \Type{i}. \mathit{list}~\Typ, \\
    & \mathit{cons} : \Forall \Typ : \Type{i}. \Typ \to \mathit{list}~\Typ \to \mathit{list}~\Typ\}\\
    \mathit{list}_{i} \eqdef{} & \mathcal{L}_{i}.\mathit{list}\\
    \mathit{nil}_{i} \eqdef{} & \mathcal{L}_{i}.\mathit{nil}\\
    \mathit{cons}_{i} \eqdef{} & \mathcal{L}_{i}.\mathit{cons}
  \end{align*}
\end{example}

\begin{example} Definition of a finitely branching tree as a mutual
  inductive block in \pCuIC{}.
  \begin{align*}
    \mathbfsf{Ind}_0\{
    &\mathit{FTree} : \Type{0}, \mathit{Forest} : \Type{0} :=\\
    & \mathit{leaf} : \mathit{FTree}, \mathit{node} : \mathit{Forset} \to \mathit{FTree} ,\\
    & \mathit{Fnil} : \mathit{Forest}, \mathit{Fcons} : \mathit{FTree} \to \mathit{Forest} \to \mathit{Forest}\}
  \end{align*}
\end{example}

\paragraph{Eliminators}
The term
$\Elim{\term}{\Inds.\indtyp_k; \vv{\termB}}{\motive_{\indtyp_1}, \dots,
  \motive_{\indtyp_l}}
{\genfun_{\constr_1},\dots,\genfun_{\constr_{l'}}}$ is the elimination
of the term $\term$ (which is of type $\indtyp_k$ (in some inductive
block $\Inds$) applied to parameters $\vv{\termB}$) where the result of
elimination of inductive types in the block, i.e., motives of
eliminations, is given by $\vv{\motive}$ and $\vv{\genfun}$ are
functions for elimination of terms constructed using particular
constructors. The term, above, has type
$\motive_{\indtyp_k}~\vv{\termB}~\vv{\vartermC}~\term$ if $\term$ has type
$\indtyp_k~\vv{\termB}~\vv{\vartermC}$.

Each case-eliminator $\genfun_{\constr_i}$ is the recipe for
generating a term of the appropriate type (according to the
corresponding motive) out of arguments of the constructor $\constr_i$
under the assumption that all (mutually) recursive arguments are
already appropriately eliminated.  This is perhaps best seen in the
rule \textsc{Iota} below which describes the judgemental equality
corresponding to the (intended) operational behavior of eliminators.

The function
$\elimtyp{\Inds}{\vv{\motive}}(\constr_i, \declcons(\constr_i))$
ascribes a type to the case-eliminator $\genfun_{\constr_i}$ in the
manner explained above. That is,
$\elimtyp{\Inds}{\vv{\motive}}(\constr_i, \declcons(\constr_i))$ is a
function type that given arguments of the constructor $\constr_i$ (and
their eliminated version if they are (mutually) recursive arguments)
produces a term of the appropriate type according to the motives. It
is formally defined as follows by recursion on derivation of
$\strictpos{\dom(\declinds)}(\declcons(\constr_i))$:
\begin{align*}
\intertext{If
  $\TypP \synteq \Forall \vv{\var} :
  \vv{\Typ}. \indtyp_i~\vv{\vartermC}$ and we have
  $\strictposarg{\dom{\declinds}}(\TypP)$ and $\strictpos{\dom{\declinds}}(\TypB)$}
  \elimtyp{\Inds}{\vv{\motive}}(\term, \TypP \to \TypB) \eqdef{} & \Forall \termP : \TypP. (\Forall \vv{\var} : \vv{\Typ}. \motive_{\indtyp_i}~\vartermC~(\termP~\vv{\var})) \to \elimtyp{\Inds}{\vv{\motive}}(\term~\termP, \TypB)\\
\intertext{Otherwise,}
  \elimtyp{\Inds}{\vv{\motive}}(\term, \Forall \var : \Typ. \TypB) \eqdef{} & \Forall \var : \Typ. \elimtyp{\Inds}{\vv{\motive}}(\term~\var, \TypB)
\intertext{Otherwise, if $\indtyp \in \dom(\declinds)$}
  \elimtyp{\Inds}{\vv{\motive}}(\term, \indtyp~\vv{\vartermC}) \eqdef{} & \motive_{\indtyp}~\vv{\vartermC}~\term
\end{align*}

\begin{example} The following is the definition of the recursive
  function corresponding to the principle of mathematical induction on
  natural numbers in \pCuIC{}.
  \begin{align*}
    \mathit{nat\_ind} \eqdef{} \Lam P : \mathit{nat} \to \Prop. \Lam o : P~zero.
    \Lam \mathit{stp} : \Forall \var : \mathit{nat}. P~\var \to P~(\mathit{succ}~\var). \\
    \Lam n : \mathit{nat}.\Elim{n}{\mathit{nat}; \nil}{P}{o, \mathit{stp}}
  \end{align*}
The term $\mathit{nat\_ind}$ above has the type
\begin{align*}
  \Forall P : \mathit{nat} \to \Prop. (P~Z) \to (\Forall \var :
  \mathit{nat}. P~\var \to P~(S~\var)) \to \Forall n : \mathit{nat}. P~n
\end{align*}
\end{example}

\begin{example} The following is the definition the recursive function
  to add two natural numbers in \pCuIC{}.
  \begin{align*}
    \mathit{add} \eqdef{} \Lam n : \mathit{nat}. \Lam m : \mathit{nat}.
    \Elim{n}{\mathit{nat}; \nil}{\Lam \_ : \mathit{nat}. \mathit{nat}}
    {m, \Lam \var : \mathit{nat}. \Lam \varB : \mathit{nat}. \mathit{succ}~\varB}
  \end{align*}
\end{example}

\begin{example} The following is the definition the polymorphic and
  universe polymorphic recursive function to compute the length of a
  list in \pCuIC{}.
  \begin{align*}
    \mathit{length} \eqdef{} & \Lam \Typ : \Type{i}. \Lam l : \mathit{list}_{i}~\Typ.
    \mathbfsf{Elim}(l; \mathit{list}; \Typ; \Lam \_ : \mathit{list}_{i}. \mathit{nat})\\
    & \set{\Lam \TypB : \Type{i}. \mathit{zero}, \Lam \TypB : \Type{i}. \Lam \var : \TypB. \Lam \varB : \mathit{list}_{i}~\TypB. \Lam \varC : \mathit{nat}. \mathit{succ}~\varC}
  \end{align*}
\end{example}

\begin{remark}
  From the definition of eliminators, it might appear that they limit
  expressivity of our system.  It might so appear that they only allow
  us to define recursive functions that are paramtrically polymorphic in
  the parameters of the inductive block. This, however is not the
  case, as existence of inductive types allows for breaking this
  paramtricity. In order to demonstrate this, the next example defines
  a function that sums up a list of natural numbers, where list is the
  polymorphic type of lists defined above.

  It is necessary to formulate eliminators such that motives and
  case-eliminators include parameters. This is due to the way
  inductive types are defined. In particular, a constructor is only
  required to produce a term of the inductive type with the parameters
  the same as the whole block. When an inductive type (including the
  one that the constructor in question is producing) appears as an
  argument of the constructor it needs not be with the same parameters
  as that of the inductive block. This is also the case in the
  practical implementation of inductive types in \Coq{}.  An example
  of such an inductive type is the following definition of lists where
  one can instead of an element of the list store \emph{some code} in
  the list representing that element as a \emph{list} of natural
  numbers.
  \begin{align*}
    \mathbfsf{Ind}_1\{&\mathit{Clist} : \Forall \Typ : \Type{i}. \Type{i} :=
                        \mathit{Cnil} : \Forall \Typ : \Type{i}. \mathit{Clist}~\Typ, \\
                      & \mathit{Ccons} : \Forall \Typ : \Type{i}. \Typ \to \mathit{Clist}~\Typ \to \mathit{Clist}~\Typ\\
                      & \mathit{Ccons'} : \Forall \Typ : \Type{i}. \mathit{Clist}~\mathit{nat} \to \mathit{Clist}~\Typ \to \mathit{Clist}~\Typ\}\\
  \end{align*}
\end{remark}

\begin{example}
  The following is the definition a function that sums up a list of
  natural numbers in \pCuIC{}. Notice that this is a
  non-paramtrically-polymorphic recursive function defined in our
  system for a polymorphic definition of lists.

  We first define an inductive type that can determine when a type is
  the type of natural numbers.
  \begin{align*}
    \mathcal{N} \eqdef{} & \mathbfsf{Ind}_0\{\mathit{isNat} : \Forall \Typ : \Type{0}. \Type{0} :=
                            \mathit{IN} : \mathit{isNat}~\mathit{nat} \}\\
    \mathit{isNat} \eqdef{} & \mathcal{N}.\mathit{isNat} \\
    \mathit{IN} \eqdef{} & \mathcal{N}.\mathit{IN} 
  \end{align*}
  The elimination of this type allows us to write a function that
  given a type $\Typ$ and an element of the type $\mathit{isNat}~\Typ$
  we can construct a function from $A$ to the natural numbers.
  \begin{align*}
    \mathit{toNat} \eqdef{} & \Lam \Typ : \Type{0}. \Lam i : \mathit{isNar}~\Typ.
    \mathbfsf{Elim}(i; \mathit{isNat}; \nil;\Lam \TypB : \Type{0}.\Lam \_ : \mathit{isNat}~\TypB. \TypB \to \mathit{nat})\\
    & \set{\Lam \var : \mathit{nat}. \var}
  \end{align*}

  Now we can use these constructions to define the basis for our sum
  function.
  \begin{align*}
    \mathit{sum\_el'} \eqdef{} & \Lam \Typ : \Type{0}. \Lam \Typ : \mathit{list}_0~\Typ.
    \mathbfsf{Elim}(l; \mathit{list}_0; \Typ; \Lam \TypB : \Type{0}. \mathit{isNar}~\Typ \to \mathit{nat})\\
    & \{\Lam \TypB : \Type{0}. \Lam \_ : \mathit{isNar}~\Typ. \mathit{zero}, \\
    & \Lam \TypB : \Type{0}. \Lam \var : \TypB. \Lam \varB : \mathit{list}~\TypB. \Lam \varC : \mathit{isNat}~\TypB \to \mathit{nat}. \Lam i : \mathit{isNar}~\Typ.\\
    & \mathit{add}~(\mathit{toNat}~i\var)~(\varC~i)\}
  \end{align*}
  Using this helper function we can define the \emph{non-polymorphic}
  function that given a list of natural numbers computes the sum of
  all its elements.
  \[\mathit{sum\_el} \eqdef{}\Lam l : \mathit{list}_0~\mathit{nat}. \mathit{sum\_el'}~\mathit{nat}~l~\mathit{IN} \]

\end{example}

\subsection{Judgemental equality}
The main typing rules for judgemental equality, except for those
related to cumulativity, are presented in Figure~\ref{fig:pCuIC:judgemental-equality}.

\begin{figure}
\begin{mathpar}
\inferH{Eq-ref}{\typed{\ctx}{\term}{\Typ}}{\jueq{\ctx}{\term}{\term}{\Typ}}
\and
\inferH{Eq-sym}{\jueq{\ctx}{\term}{\term'}{\Typ}}{\jueq{\ctx}{\term'}{\term}{\Typ}}
\and
\inferH{Eq-trans}{\jueq{\ctx}{\term}{\term'}{\Typ} \and \jueq{\ctx}
    {\term}{\term''}{\Typ}}{\jueq{\ctx}{\term}{\term''}{\Typ}}
\and
\inferH{Beta}{\typed{\ctx, \var : \Typ}{\Term}{\TypB} \and
    \typed{\ctx, \var : \Typ}{\TypB}{\sort} \and \typed{\ctx}{\TermB}{\Typ}}
  {\jueq{\ctx}{(\Lam \var : \Typ. \Term)~\TermB}
    {\subst{\Term}{\var}{\TermB}}{\subst{\TypB}{\var}{\TermB}}}
\and
\inferH{Delta}{\WFctx{\ctx} \and \var := \term : \Typ \in \ctx}
  {\jueq{\ctx}{\var}{\term}{\Typ}}
\and
\inferH{Eta}{\typed{\ctx}{\term}{\Forall \var : \Typ. \TypB}}
  {\jueq{\ctx}{\term}
    {\Lam \var : \Typ. \term~\var}{\Forall \var : \Typ. \TypB}}
\and
\inferH{Zeta}{\typed{\ctx}{\term}{\Typ} \and
    \typed{\ctx, \var := \term : \Typ}{\termB}{\TypB}}
  {\jueq{\ctx}{(\Let \var := \term : \Typ in \termB)}
    {\subst{\termB}{\var}{\term}}{\subst{\TypB}{\var}{\term}}}
\and
\inferH{Iota}{\WFctx{\ctx} \and \Inds \synteq \Ind_n\{\declinds := \declcons\} \in \ctx
  \and \subinds{\ctx}{\Inds'}{\Inds} \\\
 \dom(\declinds) = \set{\indtyp_1, \dots, \indtyp_l} \and \dom(\declcons) = \set{\constr_1, \dots, \constr_{l'}} \and \\\
\typed{\ctx}{\motive_{\indtyp_i}}{\Forall \vv{\var}:\vv{\Typ}. (\indtyp_i~\vv{\var}) \to \sort'} \text{ where }
\declinds(\indtyp_i) \synteq \Forall \vv{\var}:\vv{\Typ}.\sort \text{ for all } 1 \leq i \leq l \and \\\
\typed{\ctx}{\Inds'.\constr_j~\vv{\varterm}}{\Inds.\indtyp_k~\vv{\termB}~\vv{\vartermC}} \and \len(\vv{\termB}) = n \and
1 \le j \le l' \and\\\
\typed{\ctx}{\genfun_{\constr_i}}
{\elimtyp{\Inds}{\vv{\motive}}(\Inds.\constr_i, \declcons(\constr_i))} \text{ for all } 1 \le i \le l'
}
{\ctx \vdash {\Elim{\Inds'.\constr_j~\vv{\varterm}}{\Inds.\indtyp_k;\vv{\termB}}{\motive_{\indtyp_1}, \dots, \motive_{\indtyp_l}}
{\genfun_{\constr_1},\dots,\genfun_{\constr_{l'}}}} \simeq \\ {\recor{\Inds}{\vv{\motive}}{\vv{\genfun}}(\genfun_{\constr_j}; \vv{\varterm}; \declcons(\constr_j))} : {\motive_{\indtyp_k}~\vv{\termB}~\vv{\vartermC}~(\Inds'.\constr_j~\vv{\varterm})}}
\end{mathpar}
\caption{The main judgemental equality rules}
\label{fig:pCuIC:judgemental-equality}
\end{figure}

The rules \textsc{Eq-ref}, \textsc{Eq-sym} and \textsc{Eq-trans} make
the judgemental equality an equivalence relation. The rule
\textsc{Beta} corresponds to the operational rule for
$\beta$-reduction in lambda calculus. The rules \textsc{Delta} and
\textsc{Zeta} correspond to the operation of expansion of global
definitions (defined in the context) and let-bindings, respectively.
The rule \textsc{Eta} corresponds to $\eta$-expansion for (dependent)
functions.

The rule \textsc{Iota} corresponds to the operation of reduction of
recursive functions and case analysis (in systems featuring these,
e.g., \Coq{} itself) and that of eliminators as in our case. Notice
that this rule only applies if the term on which elimination is
being performed is a constructor applied to some terms\footnote{This
  restriction is indeed necessary in \Coq{} so as to guarantee strong
  normalization.}.  This rule specifies that, as expected, the
eliminator, applied to a term that is constructed out of a constructor
(of the corresponding type or any of its sub-types, see
Remark~\ref{remark:iota-and-ind-cum} below) applied to some terms, is
equivalent to the corresponding case-eliminator applied arguments of
the constructor, after (mutually) recursive arguments are
appropriately eliminated.  This is exactly what the recursor
$\recor{\Inds}{\vv{\motive}}{\vv{\genfun}}$ does, applying arguments
of the constructor to the corresponding case-eliminator after
eliminating (mutually) recursive arguments as necessary. Recursors are
defined as follows by recursion on the derivation of
$\strictpos{\dom(\declinds)}(\declcons(\constr_i))$:
\begin{align*}
\intertext{If
  $\TypP \synteq \Forall \vv{\var} :
  \vv{\Typ}. \indtyp_i~\vv{\vartermC}$ and we have
  $\strictposarg{\dom{\declinds}}(\TypP)$ and $\strictpos{\dom{\declinds}}(\TypB)$}
  \recor{\Inds}{\vv{\motive}}{\vv{\genfun}}(\term; \vartermB,\vv{\varterm}; \TypP \to \TypB) \eqdef{}
  & \recor{\Inds}{\vv{\motive}}{\vv{\genfun}}\left(\term~\vartermB~\left(\Lam \vv{\var} : \vv{\Typ}. \Elim{\vartermB~\vv{\var}}{\Inds.\indtyp_i; \vv{\varB}}{\vv{\motive}}{\vv{\genfun}}\right); \vv{\varterm}; \TypB\right)\\
\intertext{where $\vv{\varB}$ are the first variables in $\vv{\var}$ corresponding to parameters.}
\intertext{Otherwise,}
  \recor{\Inds}{\vv{\motive}}{\vv{\genfun}}(\term; \vartermB,\vv{\varterm}; \Forall \var : \Typ. \TypB) \eqdef{}
  & \recor{\Inds}{\vv{\motive}}{\vv{\genfun}}(\term~\vartermB; \vv{\varterm}; \TypB)
\intertext{Otherwise, if $\indtyp \in \dom(\declinds)$}
  \recor{\Inds}{\vv{\motive}}{\vv{\genfun}}(\term; \nil;\indtyp~\vv{\vartermC}) \eqdef{} & \term
\end{align*}
\begin{remark}[\textsc{Iota} and subtyping of cumulativity for inductive
  tpyes] \label{remark:iota-and-ind-cum} We, in addition, stipulate here
  that the eliminators can eliminate terms constructed using the
  corresponding constructor of any inductive type that is a sub-type of
  the inductive type for which the elimination is specified. Note that
  in the \Coq{} implementation of cumulativity for inductive types,
  cumulativity is only considered for different instances of the same
  inductive type at different universe levels. That is, only for two
  instances of the same universe polymorphic inductive type. In those
  settings, the inductive types being sup-types of one another are
  instance of the same inductive types and the eliminators, and in
  general the operational semantics, simply ignore the universes in
  terms. Here, we have to consider this side condition to achieve a
  similar result.
\end{remark}

\subsection{Cumulativity}
Figures~\ref{fig:pCuIC:cumulativity} and \ref{fig:pCuIC:cumu-jueq-inducrtive-types} show the rules for cumulativity in \pCuIC{}.
\begin{figure}
\begin{mathpar}
\inferH{Prop-in-Type}{}{\subtyp{\ctx}{\Prop}{\Type{i}}}
\and
\inferH{Cum-Type}{i \le j}{\subtyp{\ctx}{\Type{i}}{\Type{j}}}
\and
\inferH{Cum-trans}{\subtyp{\ctx}{\Term}{\TermB} \and
    \subtyp{\ctx}{\TermB}{\TermC}}{\subtyp{\ctx}{\Term}{\TermC}}
\and
\inferH{Cum-weaken}{\subtyp{\ctx}{\Term}{\TermB} \and
    \typed{\ctx}{\TermC}{\sort} \and \var \not\in\dom(\ctx)}
  {\subtyp{\ctx, \var : \TermC}{\Term}{\TermB}}
\and
\inferH{Cum-Prod}{\jueq{\ctx}{\Typ_1}{\TypB_1}{\sort} \and
    \subtyp{\ctx, \var : \Typ_1}{\Typ_2}{\TypB_2}}
  {\subtyp{\ctx}{\Forall \var : \Typ_1. \Typ_2}{\Forall \var : \TypB_1. \TypB_2}}
\and
\inferH{Cum-eq-L}{\jueq{\ctx}{\Term}{\TermB}{\sort} \and
    \subtyp{\ctx}{\TermB}{\TermC}}
  {\subtyp{\ctx}{\Term}{\TermC}}
\and
\inferH{Cum-eq-R}{\subtyp{\ctx}{\Term}{\TermB} \and
    \jueq{\ctx}{\TermB}{\TermC}{\sort}}
  {\subtyp{\ctx}{\Term}{\TermC}}
\and
\inferH{Cum}{\typed{\ctx}{\term}{\Typ} \and \subtyp{\ctx}{\Typ}{\TypB}}
  {\typed{\ctx}{\term}{\TypB}}
\and
\inferH{Cum-eq}{\jueq{\ctx}{\term}{\term'}{\Typ} \and \subtyp{\ctx}{\Typ}{\TypB}}
  {\jueq{\ctx}{\term}{\term'}{\TypB}}
\and
\inferH{Eq-Cum}{\jueq{\ctx}{\Term}{\Term'}{\sort}}{\subtyp{\ctx}{\Term}{\Term'}}
\end{mathpar}
\caption{Cumulativity}
\label{fig:pCuIC:cumulativity}
\end{figure}
Rules in Figure~\ref{fig:pCuIC:cumu-jueq-inducrtive-types} pertain to
cumulativity of inductive types. These are novel additions to the
predicative calculus of constructions (\pCIC{}). The main purpose of
the model constructed in this paper is to prove these rules correct.
\begin{figure}
\begin{mathpar}
\inferH{Ind-leq}{\Inds \synteq \Ind_n\{\declinds := \declcons\} \in \ctx \and
    \Inds' \synteq \Ind_n\{\declinds' := \declcons'\} \in \ctx \and \\\
    \dom(\declinds) = \dom(\declinds') \and
    \dom(\declcons) = \dom(\declcons') \and \\\
    \Big[\declinds(\indtyp) \synteq \vv{\termP} : \vv{\TypP}.
       \Forall \vv{\varC} : \vv{\varTypC}.\sort \and
    \declinds'(\indtyp) \synteq \vv{\termP} : \vv{\TypP'}.
       \Forall \vv{\varC} : \vv{\varTypC'}.\sort' \and \\\
    \subtyp{\ctx, \vv{\termP} : \vv{\TypP}}{\vv{\varTypC}}{\vv{\varTypC'}} \and \\\
    \Big(\declcons(\constr) \synteq \Forall \vv{\termP} : \vv{\TypP}.
       \Forall \vv{\var} : \vv{\varTypB}.\indtyp~\vv{\termB} \and
    \declcons'(\constr) \synteq \Forall \vv{\termP} : \vv{\TypP'}.
       \Forall \vv{\var} : \vv{\varTypB'}.\indtyp~\vv{\termB'} \and \\\
    \subtyp{\ctx, \vv{\termP} : \vv{\TypP}}{\vv{\varTypB}}{\vv{\varTypB'}}
    \and \jueq{\ctx, \vv{\termP} : \vv{\TypP}, \vv{\var} : \vv{\varTypB}}
           {\vv{\termB}}{\vv{\termB'}}{\vv{\TypP'}, \vv{\varTypC'}} \and \\\
    \text{for } \constr \in \constrsof(\declcons, \indtyp) \Big) \and
    \text{for } \indtyp \in \dom(\declinds) \Big]
  }{\subinds{\ctx}{\Inds}{\Inds'}}
\and
\inferH{C-Ind}{
  \Inds \synteq \Ind_n\{\declinds := \declcons\} \and
  \Inds' \synteq \Ind_n\{\declinds' := \declcons'\} \and
  \subinds{\ctx}{\Inds}{\Inds'} \and\\\
  \typed{\ctx}{\Inds.\indtyp~\vv{\varterm}}{\sort} \and
  \typed{\ctx}{\Inds'.\indtyp~\vv{\varterm}}{\sort'}
}{\subtyp{\ctx}{\Inds.\indtyp~\vv{\varterm}}{\Inds'.\indtyp~\vv{\varterm}}}
\and
\inferH{Ind-Eq}{\subtyp{\ctx}{\Inds.\indtyp~\vv{\varterm}}{\Inds'.\indtyp~\vv{\varterm}} \and
  \subtyp{\ctx}{\Inds'.\indtyp~\vv{\varterm}}{\Inds.\indtyp~\vv{\varterm}} \and
  \typed{\ctx}{\Inds.\indtyp~\vv{\varterm}}{\sort} \and
  \typed{\ctx}{\Inds'.\indtyp~\vv{\varterm}}{\sort}}
{\jueq{\ctx}{\Inds.\indtyp~\vv{\varterm}}{\Inds'.\indtyp~\vv{\varterm}}{\sort}}
\and
\inferH{Constr-Eq-L}{\subtyp{\ctx}{\Inds'.\indtyp~\vv{\varterm}}{\Inds.\indtyp~\vv{\varterm}} \and \\\
  \typed{\ctx}{\Inds.\constr~\vv{\vartermC}}{\Inds.\indtyp~\vv{\varterm}} \and
  \typed{\ctx}{\Inds'.\constr~\vv{\vartermC}}{\Inds'.\indtyp~\vv{\varterm}}}
{\jueq{\ctx}{\Inds.\constr~\vv{\vartermC}}{\Inds'.\constr~\vv{\vartermC}}{\Inds.\indtyp~\vv{\varterm}}}
\and
\inferH{Constr-Eq-R}{\subtyp{\ctx}{\Inds.\indtyp~\vv{\varterm}}{\Inds'.\indtyp~\vv{\varterm}} \and \\\
  \typed{\ctx}{\Inds.\constr~\vv{\vartermC}}{\Inds.\indtyp~\vv{\varterm}} \and
  \typed{\ctx}{\Inds'.\constr~\vv{\vartermC}}{\Inds'.\indtyp~\vv{\varterm}}}
{\jueq{\ctx}{\Inds.\constr~\vv{\vartermC}}{\Inds'.\constr~\vv{\vartermC}}{\Inds'.\indtyp~\vv{\varterm}}}
\end{mathpar}
\caption{Cumulativity and judgemental equality for inductive types}
\label{fig:pCuIC:cumu-jueq-inducrtive-types}
\end{figure}

The rule \textsc{Ind-leq} specifies when we say a blocks of inductive
types $\Inds$ is included in another block $\Inds'$. Intuitively, this
holds when corresponding arguments of the inductive type (only the
arity and not the parameters) of corresponding inductive types have
the appropriate subtyping relation and also for all corresponding
arguments of corresponding constructors, except for those that are
parameters of the inductive types. Excluding parameters has
interesting consequences. For instance, let us consider the following
example that shows the polymorphic definition of lists of elements of
a type $\Typ : \Type{i}$ are independent of the level of the type,
$i$.  Notice that cumulativity, implies that $\Typ$ also has type
$\Type{j}$ for any $i \le j$.

\begin{example}\label{ex:cum-jueq-lists}
  Consider the universe polymorphic definition of lists from Example~\ref{ex:upolylist}.
  This definitions allows us to use rules \textsc{Int-leq} and \textsc{C-Ind} to
  conclude that
  $\subtyp{\mathcal{L}_{i}, \mathcal{L}_{j}, \ctx}{\mathit{List}_{i}~\Typ}{\mathit{List}_{j}~\Typ}$ for
  all types $\Typ$ regardless of $i$ and $j$ which would not be the
  case if we were to compare parameters.  This result in turn allows
  us to use the rule \textsc{Ind-Eq} to conclude that
  $\jueq{\mathcal{L}_{i}, \mathcal{L}_{j}, \ctx}{\mathit{List}_{i}~\Typ}{\mathit{List}_{j}~\Typ}{\Type{i}}$
  and use the rule \textsc{Constr-Eq-L} to conclude
  $\jueq{\mathcal{L}_{i}, \mathcal{L}_{j}, \ctx}{\mathit{nil}_{i}~\Typ}{\mathit{nil}_{j}~\Typ}{\mathit{List}_{i}~\Typ}$
  both regardless of $i$ and $j$. Similarly, we can conclude that
  $\jueq{\mathcal{L}_{i}, \mathcal{L}_{j}, \ctx}{\mathit{cons}_{i}~\Typ~a~l}{\mathit{nil}_{j}~\Typ~a'~l'}{\mathit{List}_{i}~\Typ}$
  whenever $\jueq{\mathcal{L}_{i}, \mathcal{L}_{j}, \ctx}{l}{l'}{\mathit{List}_{i}~\Typ}$
  $\jueq{\mathcal{L}_{i}, \mathcal{L}_{j}, \ctx}{a}{~a'}{\Typ}$.
\end{example}

Notice that rules \textsc{C-Ind}, \textsc{Ind-Eq},
\textsc{Constr-Eq-L} and \textsc{Constr-Eq-R} are only applicable if
the inductive type or the constructor in the latter case are fully
applied. This is mainly because subtyping only makes sense for types
(elements of a sort).



\section{Set-theoretic background}
In this section we shortly explain the set-theoretic axioms and
constructions that form the basis of our model. We assume that the
reader is familiar with the ZFC set theory. This is very similar to
the theory that \citet{DBLP:journals/corr/abs-1111-0123} use as the
basis for their model. In particular, we use Zermelo-Fraenkel set
theory with the axiom of choice (ZFC) together with an axiom that
there is a countably infinite strictly increasing hierarchy of
uncountable strongly inaccessible cardinals. In particular, we assume
that we have a hierarchy of \emph{strongly inaccessible} cardinals
$\icard_{0}, \icard_{1}, \icard_{2}, \dots$ where
$\icard_{0} > \omega$.

\subsection{Von Neumann cumulative hierarchy and models of ZFC}
The von Neumann cumulative hierarchy is a sequence of sets (indexed by
ordinal numbers) that is cumulative. That is, each set in the
hierarchy is a subset of the all sets after it. These sets are also
referred to as von Neumann universes. This hierarchy is defined as
follows for ordinal number $\ord$:
\[
\VNCH_\ord \eqdef \bigcup_{\ordB \in \ord} \powerset{\VNCH_\ordB}
\]
It is well-known \citep{DrakeFrankR1974St:a} that whenever $\ord$ is a
strongly inaccessible cardinal number of a cardinality strictly
greater than $\omega$, as is the case for $\icard_0, \icard_2, \dots$,
$\VNCH_\ord$ is a model for ZFC. The von Neumann universe
$\VNCH_{\omega}$ satisfies all axioms of ZFC except for the axiom of
infinity. This is due to the fact that all sets in $\VNCH_{\omega}$
are \emph{finite}. This is why we have assumed that our hierarchy of
strngly inaccessible cardinals starts at one strictly greater than
$\omega$. In particular, if $A$ and $B$ are two sets in
$\VNCH_{\icard}$ then $B^A \in \VNCH_{\icard}$ where $B^A$ is the set
of all functions from $B$ to $A$. This allows us to use von Neumann
universes to model the predicative \pCuIC{} universes.  For more
details about strongly inaccessible cardinals and von Neumann
universes refer to \citet{DrakeFrankR1974St:a}.

\subsection{Rule sets and fixpoints: inductive constructions in set theory}
Following \citet{DBLP:journals/corr/abs-1111-0123}, who follow
\citet{Dybjer91inductivesets} and \citet{Aczel1999}, we use inductive
definitions (in set theory) constructed through rule sets to model
inductive types.  Here, we give a very short account of rule sets for
inductive definitions. For further details refer to
\cite{ACZEL1977739}.

A pair $(\premises, \concl)$ is a \emph{rule} based on a set
$\gensetC$ where $\premises \subseteq \gensetC$ is the set of premises
and $\concl \in \gensetC$ is the conclusion. We usually write
$\Rule{\premises}{\concl}$ for a rule $(\premises, \concl)$. A
\emph{rule set} based on $\gensetC$ is a set $\ruleset$ of
rules based on $\gensetC$. We say a set $\genset \subseteq \gensetC$
is \emph{$\ruleset$-closed}, $\closed{\ruleset}(\genset)$ for a
$\gensetC$ based rule set $\ruleset$ if we have:
\[
  \closed{\ruleset}(\genset) \eqdef \forall \Rule{\premises}{\concl} \in
  \ruleset.~\premises \subseteq \genset \then \concl \in \genset
\]
The operator $\operation{\ruleset}$ corresponding to a rule set
$\ruleset$ is the operation of collecting all conclusions for a set
whose premises are available in that set. That is,
\[
  \operation{\ruleset}(\genset) \eqdef \set{\concl \middle|
    \Rule{\premises}{\concl} \in \ruleset \land \premises \subseteq
    \genset}
\]
Hence, a set $\genset$ is $\ruleset$-closed if
$\operation{\ruleset}(\genset) \subseteq \genset$. Notice that
$\operation{\ruleset}$ is a monotone function on $\powerset{\gensetC}$
which is a complete lattice. Therefore, for any $\gensetC$ based rule
set $\ruleset$, the operator $\operation{\ruleset}$ has a least
fixpoint, $\inddef{\ruleset} \subseteq \gensetC$:
\[
  \inddef{\ruleset} \eqdef \bigcap \set{\genset \subseteq \gensetC
    \middle| \closed{\ruleset}(\genset)}
\]
We define by transfinite recursion a sequence, indexed by ordinal
numbers $\TFoperation{\ruleset}{\ord}$ for an ordinal number
$\ord$:
\[
  \TFoperation{\ruleset}{\ord} \eqdef \bigcup_{\ordB \in \ord}
  \left(\TFoperation{\ruleset}{\ordB} \cup
  \operation{\ruleset}(\TFoperation{\ruleset}{\ordB})\right)
\]
Obviously, for $\ordB \le \ord$ we have
$\TFoperation{\ruleset}{\ordB} \subseteq
\TFoperation{\ruleset}{\ord}$.

\begin{theorem}[\citet{ACZEL1977739}]
  For any rule set $\ruleset$ there exists an ordinal number
  $\closingord{\ruleset}$ called the \emph{closing ordinal} of
  $\ruleset$ such that it is the smallest ordinal number for which we
  have
  $\inddef{\ruleset} = \TFoperation{\ruleset}{\closingord{\ruleset}}$
  is the least fixpoint of $\operation{\ruleset}$. In other words,
  $\TFoperation{\ruleset}{\closingord{\ruleset}+1} =
  \operation{\ruleset}(\TFoperation{\ruleset}{\closingord{\ruleset}}) =
  \TFoperation{\ruleset}{\closingord{\ruleset}}$.
\end{theorem}

\begin{lemma}[\citet{ACZEL1977739}] \label{lem:closing-ord-bound}
  Let $\ruleset$ be a rule
  set based on some set $\gensetC$ and $\ordB$ a \emph{regular}
  cardinal such that for every rule
  $\Rule{\premises}{\concl} \in \ruleset$ we have that
  $\cardinality{\premises} < \ordB$ then
  \[
    \closingord{\ruleset} \le \ordB
  \]
  In other words, $\TFoperation{\ruleset}{\ordB}$ is the least
  fixpoint of $\operation{\ruleset}$.
\end{lemma}

\subsection{Fixpoints of large functions}
A set theoretic constructions is called \emph{large} with respect to a
set theoretic (von Neumann) universe if it does not belong to that
universe. As we shall see, the functions that we will consider for
interpreting of inductive types (operators of certain rule sets)
are indeed large. That is, they map subsets of the universe to subsets
of the universe. As a result, the fixpoint of these functions might
not have a fixpoint within the universe in question as the universe
with the subset relation on it is not a complete lattice. The following
lemmas show that under certain conditions, the fixpoint of the
function induced by rule sets does exist in the desired universe. We
will use this lemma to show that the interpretations of inductive
types do indeed fall in the universe that they are supposed to.

\begin{lemma}\label{lem:closing-ord-in-universe}
  Let $\ruleset$ be a rule set based on the set-theoretic universe
  $\VNCH{}$ and $\ord \in \VNCH{}$ be a cardinal number such that for
  all $\Rule{\premises}{\concl} \in \ruleset$ we have
  $\cardinality{\premises} \le \ord$. Then,
  \[\closingord{\ruleset} \in \VNCH{}\]
\end{lemma}

\begin{proof}
  By Lemma~\ref{lem:closing-ord-bound}, it suffices to show that there
  is a regular cardinal $\ordB \in \VNCH{}$ such that
  $\cardinality{\premises} < \ordB$ for any
  $\Rule{\premises}{\concl} \in \ruleset$. Take $\ordB$ to be
  $\Aleph{\ord+1}$. By the fact that $\VNCH{}$ is a model of ZFC we
  know that $\Aleph{\ord+1} \in \VNCH$. By the fact that
  $\ord < \Aleph{\ord+1}$ we know that
  $\cardinality{\premises} < \Aleph{\ord+1}$ for any
  $\Rule{\premises}{\concl} \in \ruleset$. It is well known that under
  axioms of ZFC (this crucially uses axiom of choice) $\Aleph{\ord+1}$
  is a regular cardinal for any ordinal number $\alpha$ -- see
  \cite{DrakeFrankR1974St:a} for a proof. This concludes our proof.
\end{proof}

\begin{lemma}\label{lem:fixpoint_of_ruleset_in_universe}
  Let $\ruleset$ be a rule set based on the set-theoretic universe
  $\VNCH{}$ and $\ord \in \VNCH{}$ be a cardinal number such that for
  all $\Rule{\premises}{\concl} \in \ruleset$ we have
  $\cardinality{\premises} \le \ord$. Then,
  \[\inddef{\ruleset} \in \VNCH{}\]
\end{lemma}

\begin{proof}
  By Lemma~\ref{lem:closing-ord-bound} we know that
  $\inddef{\ruleset} =
  \TFoperation{\ruleset}{\closingord{\ruleset}}$. We know that
  $\TFoperation{\ruleset}{\closingord{\ruleset}} \in \VNCH{}$ as it is
  constructed by transfinite recursion up to $\closingord{\ruleset}$
  and that we crucially know that $\closingord{\ruleset} \in \VNCH{}$
  by Lemma~\ref{lem:closing-ord-in-universe}.  More precisely, this
  can be shown using transfinite induction up to
  $\closingord{\ruleset}$ showing that every stage belongs to
  $\VNCH{}$. Notice that it is crucial for an ordinal to belong to the
  universe in order for transfinite induction to be valid.
\end{proof}

\subsection{The use of axiom of choice}
The only place in this work that we make use of axiom of choice is in
the proof of Lemma~\ref{lem:closing-ord-in-universe}. We use this
axiom to show the following statement which we could have
alternatively taken as a (possibly) weaker axiom.
\begin{quote}
In any von Neumann universe $\VNCH{}$ for any
cardinal number $\ord$ there is a \emph{regular} cardinal $\ordB$ such
that $\ord < \ordB$.
\end{quote}

Note that his statement is independent of ZF and certain axiom, e.g.,
choice as we have taken here, needs to be postulated. This is due to
the well-known fact proven by \citet{Gitik1980} that under the
assumption of existence of strongly compact cardinals, any uncountable
cardinal is singular!

\subsection{Modeling the impredicative sort $\Prop$: trace encoding}
One of the challenges in constructing a model for a system like
\pCuIC{} is treatment of an impredicative proof-irrelevant sort
$\Prop$.  This can be done by simply modeling $\Prop$ as the set
$\set{\emptyset, \set{\emptyset}}$ where provable propositions are
modeled as $\set{\emptyset}$ and non-provable propositions as
$\emptyset$.  This however, will only work where we don't have the
cumulativity relation between $\Prop$ and $\Type{i}$. In presence of
such cumulativity relations, such a na\"ive treatment of $\Prop$
breaks interpretation of the (dependent) function spaces as sets of
functions. The following example should make the issue plain.

\begin{example}[\citet{Werner1997}]\label{ex:prop_cumul_problem}
  Let's consider the interpretation
  of the term $I \synteq \Lam (\TypP : \Type{0}). \TypP \to \TypP$.
  In this case, the semantics of $\sem{I~\mathit{True}}$ will be
  $\set{\emptyset}$ or $\set{\emptyset}^{\set{\emptyset}}$ depending
  on whether $\mathit{True} : \Prop$ or $\mathit{True} : \Type{0}$ is
  considered, respectively. And hence we should have
  $\set{\emptyset} = \set{\emptyset}^{\set{\emptyset}}$ which is not
  the case, even though the two sets are isomorphic (bijective).
\end{example}

In order to circumvent this issue, we follow
\citet{DBLP:journals/corr/abs-1111-0123}, who in turn follow
\citet{Aczel1999}, in using the method known as \emph{trace encoding}
for representation of functions.

\begin{definition}[Trace encoding]
  The following two functions, $\encode$ and $\decode$, are used for
  trace encoding and application of trace encoded functions
  respectively.
\begin{align*}
  \encode(\genfun) \eqdef{} & \bigcup_{(\var, \varB) \in \genfun}\left(\set{\var} \times \varB \right)\\
  \decode(f, \var) \eqdef{} & \set{\varB \middle| (\var, \varB) \in \genfun}
\end{align*}
\end{definition}

\begin{lemma}[]
  Let $\genfun : \genset \to \gensetB$ be a set theoretic function
  then for any $\var \in \genset$ we have
  \[
    \decode(\encode(\genfun), \var) = \genfun(\var)
  \]
\end{lemma}

Note that using the trace encoding technique the problem mentioned in
Example~\ref{ex:prop_cumul_problem} is not present anymore. That is, we have:
\[
\set{\encode(\genfun) \middle| \genfun \in \set{\emptyset}^{\set{\emptyset}}} = \set{\emptyset}
\]

\begin{lemma}[\citet{Aczel1999}]\label{lem:impred-encoding}
  Let $A$ be a set and assume the set $B(x) \subseteq 1$ for
  $x \in A$.
  \begin{enumerate}
  \item \label{lem:impred-encoding:case1}
    $\set{\encode(f) \middle| f \in \DepFun x in A. B(x)} \subseteq 1$
  \item \label{lem:impred-encoding:case2}
    $\set{\encode(f) \middle| f \in \DepFun x in A. B(x)} = 1$ iff
    $\forall x \in A.\; B(x) = 1$
  \end{enumerate}
\end{lemma}

\section{Set-theoretic model and consistency of \pCuIC{}}
We construct a model for \pCuIC{} by interpreting predicative
universes using von Neumann universes and $\Prop$ as
$\set{0, 1} = \set{\emptyset, \set{\emptyset}}$. We use the trace
encoding technique presented earlier for (dependent) function
types. We will construct the interpretation of inductive types and
their eliminators using rule sets for inductive definitions in set
theory. We shall first define a $\size$ function on terms, typing
contexts, and pairs of a context and a term which we write as
$\size(\ctx \vdash \term)$. We will then define the interpretation of
typing contexts and terms (in appropriate context) by well-founded
recursion on their size.

\begin{definition}
  We define a function $\size$ on terms, contexts, declarations and
  pairs consisting of a context and a term (which we write as
  $\size(\ctx \vdash \term)$) mutually recursively as follows:
  \begin{align*}
    \intertext{Size for typing contexts and declarations}
    \size(\cdot) \eqdef{}
    & \sfrac{1}{2}\\
    \size(\ctx, \var : \Typ) \eqdef{}
    & \size(\ctx) + \size(\Typ)\\
    \size(\ctx, \var := \term: \Typ) \eqdef{}
    & \size(\ctx) + \size(\term) + \size(\Typ)\\
    \size(\ctx, \Ind_{n}\{\declinds := \declcons\}) \eqdef{}
    & \size(\ctx) + \size(\Ind_{n}\{\declinds := \declcons\})\\
    \intertext{Size for term}
    \size(\Prop) \eqdef{} & 1\\
    \size(\Type{i}) \eqdef{} & 1\\
    \size(\var) \eqdef{} & 1\\
    \size(\Forall \var : \Typ. \TypB) \eqdef{}
    & \size(\Typ) + \size(\TypB) + 1\\
    \size(\Lam \var : \Typ. \term) \eqdef{}
    & \size(\Typ) + \size(\term) + 1\\
    \size(\term~\termB) \eqdef{}
    & \size(\term) + \size(\termB) + 1\\
    \size(\Let \var := \term : \Typ in \termB) \eqdef{}
    & \size(\term) + \size(\termB) + \size(\Typ) + 1\\
    \size(\Inds.\varC) \eqdef{}
    & \size(\Inds)\\
    \size(\Elim{\term}{\Inds.d_i; \vv{\termB}}{\vv{\motive}}{\vv{\genfun}}) \eqdef{}
    & \size(\term) + \size(\Inds) + \sum_{i} \size(\termB_i) + \sum_{i} \size(\motive_i) + \sum_{i} \size(\genfun_i) + 1\\
    \intertext{Size for pairs consisting of a context and a term}
    \size(\Ind_{n}\{\declinds := \declcons\}) \eqdef{}
    & \sum_{\indtyp \in \dom(\declinds)} \size(\declinds(\indtyp))
      + \sum_{\constr \in \dom(\declcons)} \size(\declcons(\constr)) + 1\\
    \intertext{Size for judgements}
    \size(\ctx \vdash \term) \eqdef{}
    & \size(\ctx) + \size(\term) - \sfrac{1}{2}\\
    \size(\ctx \vdash \Ind_{n}\{\declinds := \declcons\}) \eqdef{}
    & \size(\ctx) + \size(\Ind_{n}\{\declinds := \declcons\}) - \sfrac{1}{2}\\
  \end{align*}
\end{definition}

In the definition above, which is similar to that by
\citet{DBLP:journals/corr/abs-1111-0123} and \citet{Miquel2003}, the
$\pm\sfrac{1}{2}$ is add to make sure
$\size(\ctx \vdash \term) < \size(\ctx, \var : \term)$ and that
$\size(\ctx) < \size(\ctx \vdash \term)$.

\begin{definition}[The model]
  We define the interpretations for contexts and terms by well-founded
  recursion on their sizes.
  \begin{align*}
    \intertext{Interpretation of contexts}
    \sem{\cdot} \eqdef{}& \set{\nil}\\
    \sem{\ctx, \var : \Typ} \eqdef{}& \set{\env, \genel \middle| \env \in \sem{\ctx} \land \Defined{\semtyped{\ctx}{\Typ}{\env}} \land \genel \in \semtyped{\ctx}{\Typ}{\env}}\\
    \sem{\ctx, \var := \term : \Typ} \eqdef{}& \set{\env, \genel \middle| \env \in \sem{\ctx} \land \Defined{\semtyped{\ctx}{\Typ}{\env}} \land \Defined{\semtyped{\ctx}{\term}{\env}} \land \genel = \semtyped{\ctx}{\term}{\env} \in \semtyped{\ctx}{\Typ}{\env}}\\
    \sem{\ctx, \Ind_{n}\{\declinds := \declcons\}} \eqdef{}& \sem{\ctx} \text{ if } \Defined{\semtyped{\ctx}{\Ind_{n}\{\declinds := \declcons\}}{\env}} \text{ for all } \env \in \sem{\ctx}\\
    \intertext{Above, we assume that $\var \not\in \dom(\ctx)$, otherwise, both $\sem{\ctx, \var : \Typ}$ and $\sem{\ctx, \var := \term : \Typ}$ are undefined.}
    \intertext{Interpretation of terms}
    \semtyped{\ctx}{\Prop}{\env} \eqdef{}& \set{\emptyset, \set{\emptyset}}\\
    \semtyped{\ctx}{\Type{i}}{\env} \eqdef{}& \VNCH_{\icard_{i}}\\
    \semtyped{\ctx}{\var}{\vv{\genel}} \eqdef{}& \genel_{\len(\ctx_1) - l} \hspace{2em} \text{if } \ctx = \ctx_1, \var : \Typ, \ctx_2 \text{ and } \var \not\in \dom(\ctx_1) \cup \dom(\ctx_2) \\
                        & \text{ and } l = \len(\IndsOf(\ctx_1))\\
    \semtyped{\ctx}{\Forall \var : \Typ. \TypB}{\env} \eqdef{}& \set{\encode({\genfun}) \middle| f : \DepFun \genel in \semtyped{\ctx}{\Typ}{\env}.\semtyped{\ctx, \var : \Typ}{\TypB}{\env, \genel}}\\
    \semtyped{\ctx}{\Lam \var : \Typ. \term}{\env} \eqdef{}& \encode\left(\set{(\genel, \semtyped{\ctx, \var : \Typ}{\term}{\env, \genel}) \middle| \genel \in \semtyped{\ctx}{\Typ}{\env}}\right)\\
    \semtyped{\ctx}{\term~\termB}{\env} \eqdef{}& \decode(\semtyped{\ctx}{\term}{\env}, \semtyped{\ctx}{\termB}{\env})\\
    \semtyped{\ctx}{\Let \var := \term : \Typ in \termB}{\env} \eqdef{}& \semtyped{\ctx, \var := \term : \Typ}{\termB}{\env, \semtyped{\ctx}{\termB}{\env}}\\
    \intertext{Interpretation of inductive types, constructors and eliminators is defined below}
  \end{align*}
\end{definition}

\subsection{Modeling inductive types}
We define the interpretation for inductive blocks by constructing a
rule set which will interpret the whole inductive block. We will
define interpretation of individual inductive types and constructors
of the block based on the fixpoint of this rule set.

\begin{definition}[Interpretation of the inductive types]\label{def:interp_inductive_types}
  Let $\Inds \synteq \Ind_{n}\{\declinds := \declcons\}$ be an
  arbitrary but fixed inductive block such that
  $\declinds = \indtyp_0 : \Typ_0, \dots \indtyp_l : \Typ_l$ and
  $\declcons = \constr_0 : \varTyp_0, \dots \constr_{l'} :
  \varTyp_{l'}$.  Where
  $\Typ_i \synteq \Forall \vv{\termP} : \vv{\TypP}. \Forall \vv{\var}
  : \vv{\TypB_i}. \sort$ for some sequence of types $\vv{\TypB}$ and
  some sort $\sort$. Here, $\vv{\TypP}$ are parameters of the
  inductive block. The type of constructors are of the following form:
  $\varTyp_k \synteq \Forall \vv{\termP} : \vv{\TypP}. \Forall
  \vv{\var} : \vv{\varTypC_k}. \indtyp_{i_k}~\vv{\termP}~\vv{\term_k}$
  for some $\vv{\term_k}$. Notice that $\vv{\varTypC_k}$ is a sequence
  itself. That is, it is of the form
  $\vv{\varTypC_k} \synteq \varTypC_{k, 1}, \dots, \varTypC_{k,
    {\len(\vv{\varTypC_K})}}$. That is, each constructor $\constr_k$,
  takes a number of arguments, first parameters $\vv{\TypP}$ and then
  some more $\vv{\varTypC_k}$. The strict positivity condition implies
  that for any non-parameter argument $\varTypC_{k, i}$ of a
  constructor $\constr_k$, either $\indtyp \in \dom(\declinds)$ does
  not appear in $\varTypC_{k, i}$ or we have
  $\varTypC_{k, i} \synteq \Forall \vv{\var} :
  \vv{\varTypD_{k,i}}. \indtyp_{I_{k, i}}~\vv{\termQ}~\vv{\term}$
  where $\len(\vv{\termQ}) = \len(\vv{\termP})$. That is, each
  argument of a constructor where an inductive type (of the same
  block) appears, is a (dependent) function with codomain being that
  inductive type. In this case, no inductive type of the block appears
  the domain(s) of the function, $\vv{\varTypD_{k,i}}$. Notice that the
  codomain of the function is an inductive type in the block but not
  necessarily of the same family as the one being defined -- the
  parameters applied are $\vv{\termQ}$ instead of $\vv{\termP}$!  We
  write $\recarg(\varTypC_{k, i})$ if some inductive of the block
  appears in $\varTypC_{k, i}$ in which case it will be of the form
  just described.

  \begin{align*}
    \semtyped{\ctx}{\Inds}{\env} \eqdef{}
    & \inddef{\ruleset_{\ctx, \Inds}^{\env}}\\
    \ruleset_{\ctx, \Inds}^{\env} \eqdef{}
    & \bigcup_{\indtyp_i \in
      \dom(\declinds)}\bigcup_{\constr_k \in
      \constrsof(\indtyp_i)}
      \set{ \Rule{\Psi_{\indtyp_i, \constr_k}}{\psi_{\indtyp_i, \constr_k}}  \middle|
      \begin{array}{l}
        \vv{\varterm} \in
        \semtyped{\ctx}{\vv{\TypP}}{\env},\\[0.2em]
        \vv{\vartermC} \in \semtyped{\ctx, \vv{\termP} : \vv{\TypP}, \vv{\indtyp'} : \vv{\sort}}{\vv{\varTypC_k'}}{\env, \vv{\varterm}, \vv{\vartermB}}
      \end{array}
    }
    \intertext{Where, $\vv{\indtyp'}$ is the sequence of $\indtyp_{I_{k, j}}'$'s occurring within $\vv{\varTypC'_{k}}$ below and $\sort_i \synteq \arity_{\indtyp_i}$ is the arity of the inductive type $\indtyp_i$, $\vv{\vartermB} \in \semtyped{\ctx, \vv{\termP} : \vv{\TypP}}{\vv{\sort}}{\env, \vv{\varterm}}$ and,}
    \psi_{\indtyp_i, \constr_k} \eqdef{}
    & \tuple{i;
      \vv{\varterm}; \semtyped{\ctx, \vv{\termP} : \vv{\TypP},
      \vv{\indtyp'} : \vv{\sort}, \vv{\var} : \vv{\varTypC_k'}}{\vv{\term_k}}{\env,
      \vv{\varterm}, \vv{\vartermB}, \vv{\vartermC}}; \tuple{k; \vv{\vartermC}}}\\
    \Psi_{\indtyp_i, \constr_k} \eqdef{}
    &
      \bigcup_{\recarg(\varTypC_{k,j})}\set{
      \begin{array}{l}
        \left< I_{k, j}; \right. \\
        \semtyped{\ctx_{k, j}, \vv{\varB} : \vv{\varTypD_{k, j}}}{\vv{\termQ}}{\env,
        \vv{\varterm}. \vv{\vartermB}, \vv{\vartermD}, \vv{\termB}};  \\
        \semtyped{\ctx_{k, j}, \vv{\varB} : \vv{\varTypD_{k, j}}}{\vv{\term}}{\env,
        \vv{\varterm}, \vv{\vartermB}, \vv{\vartermD}, \vv{\termB}};\\
        \left. \vv{\decode}(\vartermC_{I_{k, j}}, \vv{\termB}) \right>
      \end{array}
    \middle|
    \vv{\termB} \in \semtyped{\ctx_{k,j}}{\vv{\varTypD_{k, j}}}{\env, \vv{\varterm}, \vv{\vartermB}, \vv{\vartermD}}}\\
    \ctx_{k,j} \eqdef{} & \ctx, \vv{\termP} : \vv{\TypP}, \vv{\indtyp'} : \vv{\sort}, \var_1 : \varTypC_{k, 1}', \dots, \var_{j-1} : \varTypC_{k, j-1}' \\
    \vv{\vartermD} \eqdef{}& \vartermC_1, \dots \vartermC_{j-1}\\
    \intertext{For $\Psi_k$ we assume that $\varTypC_{k, j} \synteq \Forall \vv{\varB} :
    \vv{\varTypD_{k,j}}. \indtyp_{I_{k, j}}~\vv{\termQ}~\vv{\term}$.}
    \intertext{The types $\varTypC'_{k,j}$ are defined based on $\varTypC_{k,j}$ as follows: }
    \varTypC'_{k,j} \eqdef{}
    & \begin{cases} \Forall \vv{\var} :
      \vv{\varTypD_{k,i}}. \indtyp_{I_{k, j}}' & \text{if } \recarg(\varTypC_{k,j}) \text{ and } \varTypC_{k, i} \synteq \Forall \vv{\var} :
      \vv{\varTypD_{k,i}}. \indtyp_{I_{k,j}}~\vv{\termQ}~\vv{\term} \\ \varTypC_{k, j} & \text{otherwise} \end{cases}\\
    \intertext{We define the interpretation of the inductive types in the block as follows:}
    \semtyped{\ctx}{\Inds.\indtyp_i}{\env} \eqdef{}& \vv{\encode}(\genfun_{\indtyp_i})\\
    \genfun_{\indtyp_i}(\vv{\varterm},\vv{\term}) \eqdef{}
    & \set{\tuple{k;\vv{\vartermC}} \middle| \tuple{i; \vv{\varterm}; \vv{\term};\tuple{k; \vv{\vartermC}}} \in \semtyped{\ctx}{\Inds}{\env}}\\
    & \text{ for } \vv{\varterm}, \vv{\term} \in \semtyped{\ctx}{\vv{\TypP}, \vv{\TypB_i}}{\env}\\
    \intertext{We define the interpretation of the constructors in the block as follows:}
    \semtyped{\ctx}{\Inds.\constr_k}{\env} \eqdef{}& \vv{\encode}(\genfunB_{\constr_k})\\
    \genfunB_{\constr_k}(\vv{\varterm},\vv{\vartermC}) \eqdef{}
    & \tuple{k;\vv{\vartermC}} \text{ for } \vv{\varterm}, \vv{\vartermC} \in \semtyped{\ctx}{\vv{\TypP}, \vv{\varTypC_{k}}}{\env}
  \end{align*}
\end{definition}

Let us first discuss the intuitive construction of
Definition~\ref{def:interp_inductive_types} above. In this definition
the most important part is the interpretation of the inductive block.
We have defined this as the fixpoint of the rule sets corresponding to
constructors of the inductive type. This rule set basically spells out
the following. Take, some set $\indtyp_{I_{k, j}}'$ in the universe as
a candidate for the interpretation of $j$\textsuperscript{th}
occurrence of inductive type $\indtyp_{I_{k, j}}$ in the
$k$\textsuperscript{th} constructor. This candidate,
$\indtyp_{I_{k, j}}'$, is taken to be a candidate element of the
inductive type $\indtyp_{I_{k, j}}$ in case
$\vv{\varTypD_{k,i}} = \nil$, i.e., intuitively, if the recursive
occurrence \emph{embeds} an element of $\indtyp_{I_{k, j}}$ in the
type being constructed. On the other hand, if the recursive occurrence
of $\indtyp_{I_{k, j}}$ is, intuitively, \emph{embedding} a function
with codomain $\indtyp_{I_{k, j}}$ into the type being constructed
then $\indtyp_{I_{k, j}}'$ is to be understood as the codomain of the
function being embedded by the constructor.  In each of these two
cases, we need to make sure the candidate element and or function is
\emph{correctly chosen}, i.e., we need to make sure that the element
or the \emph{range} of the function chosen is indeed in the
interpretation of the inductive type. This is where the rule sets come
to play, so to speak.  The premise set $\Psi_{k}$ makes sure that all
candidate recursive occurrences are indeed correctly chosen. This is
done by making sure that for any of the arguments of the function
being embedded (here an element is treated as a function with no
arguments!) the result of applying the candidate function to the
arguments is indeed in the interpretation of the corresponding
inductive type. Do notice that $\vv{\decode}(\genel, \nil) = \genel$.

\begin{example}[Rule set for construction of natural numbers]
\[
  \ruleset_{\ctx, \mathcal{N}}^{\env} =
  \set{\Rule{\emptyset}{\tuple{0; \nil;
        \nil;\tuple{0; \nil}}}} \cup
  \set{\Rule{\set{\tuple{0; \nil;
          \nil;\genel}}}{\tuple{0; \nil;
        \nil;\tuple{1; \genel}}} \middle| \genel \in
    \VNCH{\icard_0}}
\]
\end{example}

\begin{example}[Rule set for construction of lists]
\[
  \ruleset_{\ctx, \mathcal{L}_i}^{\env} =
  \set{\Rule{\emptyset}{\tuple{0; \Typ;
        \nil;\tuple{0;\nil}}}} \cup
  \set{\Rule{\set{\tuple{0; \Typ; \nil,
          \genelB}}}{\tuple{0; \Typ;
        \nil;\tuple{1; \genel, \genelB}}} \middle| \genelB \in
    \VNCH{\icard_0} \land \genel \in \semtyped{\ctx}{\Typ}{\env}}
\]
\end{example}

\begin{lemma}\label{lem:inerp_inductive_indices_unique}
  Values for arguments of arities are uniquely determined by the
  values for arguments each constructor (including values for
  parameters) in $\TFoperation{\ruleset_{\ctx, \Inds}^{\env}}{\ord}$
  and in particular in $\semtyped{\ctx}{\Inds}{\env}$. That is, if
  \[
    \tuple{i; \vv{\varterm}; \vv{\term};
      \tuple{k;\vv{\vartermC}}}, \tuple{i;
      \vv{\varterm}; \vv{\term'};
      \tuple{k;\vv{\vartermC}}} \in
    \TFoperation{\ruleset_{\ctx, \Inds}^{\env}}{\ord}\] then,
  $\vv{\term} = \vv{\term'}$. Analogously for
  $\semtyped{\ctx}{\Inds}{\env}$ we have that if
  \[
    \tuple{i; \vv{\varterm}; \vv{\term};
      \tuple{k;\vv{\vartermC}}}, \tuple{i;
      \vv{\varterm}; \vv{\term'};
      \tuple{k;\vv{\vartermC}}} \in
    \semtyped{\ctx}{\Inds}{\env}\] then,
  $\vv{\term} = \vv{\term'}$.
\end{lemma}

\begin{proof}
  This immediately follows from the fact that for any two rules
  \[\Rule{\genset}{\tuple{i; \vv{\varterm}; \vv{\term};
      \tuple{k;\vv{\vartermC}}}} ~\text{ and }~
  \Rule{\genset'}{\tuple{i; \vv{\varterm}; \vv{\term'};
      \tuple{k;\vv{\vartermC}}}}\] in
  $\ruleset_{\ctx, \Inds}^{\env}$ we have $\vv{\term} = \vv{\term'}$.
\end{proof}

We shall show that the interpretation of the inductive types in a
block are each in the universe corresponding to their arity. Notice
that whenever two inductive types appear in one another the syntactic
criteria for typing enforce that they are both have the same
arity.\footnote{Note that this is the case in our work as there are no
  inductive types in $\Prop$.}  Therefore, we assume without loss of
generality that all inductive types in the block have the same
arity. In case it is not the case, it must be that there are some
inductive types in the block that are not necessarily mutually
inductive with the rest of the block. Hence, those inductive types
(and their interpretations) can be considered prior to considering the
block as a whole. Therefore, in the following theorem, we assume,
without loss of generality that the all of the inductive types of a
block are of the same arity.

\begin{remark}
  We have considered inductive types with heterogeneous parameters in
  this work. However, we can only prove the fact that each inductive
  type is in the set-theoretic universe corresponding to its arity for
  inductive types with uniform parameters, i.e., inductive types where
  parameters are fixed for the whole mutual inductive block.
\end{remark}

\begin{lemma}\label{lem:interp_inductive_correct}
  Let $\Inds \synteq \Ind_{n}\{\declinds := \declcons\}$ be a block of
  inductive types with \emph{uniform} parameters such that all
  inductive types of arity $\Type{i}$. Furthermore, let us assume that
  all the terms (and particularly types) appearing in the body of the
  block are well defined under the context $\ctx$ and environment
  $\env$ and interpretation of each of these terms (and types) is in
  the interpretation of the type (correspondingly sort) that is
  expected based on the typing derivation.\footnote{These conditions
    will hold by induction hypotheses when this lemma is used in
    practice.} Then, $\Defined{\semtyped{\ctx}{\Inds}{\env}}$,
  $\semtyped{\ctx}{\Inds}{\env} \in \VNCH{\icard_j}$, where $\Type{j}$
  is the maximal \emph{sort} of the inductive types in the block, and
  $\semtyped{\ctx}{\Inds.\indtyp_j}{\env} \in
  \semtyped{\ctx}{\declinds(\indtyp_j)}{\env}$.
\end{lemma}
\begin{proof}
  The construction of $\semtyped{\ctx}{\Inds}{\env}$ depends only on
  the interpretation of terms $\semtyped{\ctx}{\termB}{\env}$ where
  $\termB$ appears in $\declinds$ or $\declcons$ and by our
  assumptions these are all defined. Therefore, we can easily see that
  $\Defined{\semtyped{\ctx}{\Inds}{\env}}$.  We show that
  $\semtyped{\ctx}{\Inds}{\env} \in \VNCH{\icard_j}$. This follows by
  Lemma~\ref{lem:fixpoint_of_ruleset_in_universe}. Notice that as all
  terms appearing in the inductive block have types that are in
  $\Type{j}$ we know that the cardinality of premises of rules in the
  rule set for constructing $\semtyped{\ctx}{\Inds}{\env}$ are also
  all in $\VNCH{\icard_j}$. Since there are finitely many such terms
  we can take the maximum of these cardinalities which allows us to
  use Lemma~\ref{lem:fixpoint_of_ruleset_in_universe}.

  Let $\vv{\genelB}$ be a sequence of sets that are in the
  interpretation for parameters of the mutual inductive block.  Let
  $\mathcal{F}(\ruleset) = \set{\concl = \tuple{i; \vv{\genelB};
      \vv{\term};c} \middle| \Rule{\premises}{\concl} \in \ruleset}$
  and
  $\mathcal{G}(A) = \set{a = \tuple{i; \vv{\genelB}; \vv{\term};c}
    \middle| a \in A}$. That is, $\mathcal{F}$ filters a rule set so
  that only those rules are retained that produce inductive types in
  the family indexed by parameter values $\vv{\genelB}$. Similarly,
  $\mathcal{G}$ filters the fixpoint of such a rule set.

  We show, by transfinite induction up to the closing ordinal of
  $\ruleset_{\ctx, \Inds}^{\env}$, that
  \[\mathcal{G}(\inddef{\ruleset_{\ctx, \Inds}^{\env}}) = \inddef{\mathcal{F}(\ruleset_{\ctx, \Inds}^{\env})}\]
  Notice that this crucially depends on the fact that parameters are
  uniform. Now, non-parameter arguments of constructors need to be in
  the sort (see the typing rule for inductive types). Therefore, for
  each fixed set of parameters, e.g., $\vv{\genelB}$ above, for each
  constructor there is a fixed cardinality $\ord \in \VNCH{\icard_i}$
  (cardinality corresponding to the type of non-parameter arguments)
  such that the cardinality of premises of rules in
  $\mathcal{F}(\ruleset_{\ctx, \Inds}^{\env})$ corresponding to that
  constructor is less than or to $\ord$. Since, there are finitely
  many such cardinalities we can take the maximum of these
  cardinalities which is again in $\VNCH{\icard_i}$. Hence, the
  closing ordinal of
  $\inddef{\mathcal{F}(\ruleset_{\ctx, \Inds}^{\env})}$ is in
  $\VNCH{\icard_i}$. Notice, that this \emph{does not} imply that
  $\inddef{\mathcal{F}(\ruleset_{\ctx, \Inds}^{\env})}$ is in
  $\VNCH{\icard_i}$ as there are parameters can values for indices in
  the tuples in $\inddef{\mathcal{F}(\ruleset_{\ctx, \Inds}^{\env})}$.

  Finally, we show that for each sequence $\vv{c}$ of sets that are in
  the interpretation for the indices of an inductive type $\indtyp_i$,
  $\set{a \middle| \tuple{i; \vv{\genelB}; \vv{c}; a}
    \inddef{\mathcal{F}(\ruleset_{\ctx, \Inds}^{\env})}} \in
  \VNCH{\icard_i}$. We show this by transfinite induction up to the
  closing ordinal of
  $\inddef{\mathcal{F}(\ruleset_{\ctx, \Inds}^{\env})}$ which,
  crucially, we know that is in $\VNCH{\icard_i}$.
\end{proof}

\begin{lemma} \label{lem:c-ind-in-model}
  Let $\Inds \synteq \Ind_{n}\{\declinds := \declcons\}$ and
  $\Inds' \synteq \Ind_{n}\{\declinds' := \declcons'\}$ be two blocks
  of inductive types with $\dom(\declinds) = \dom(\declinds')$ and
  $\dom(\declcons) = \dom(\declcons')$. Furthermore, assume that the
  for each $\indtyp \in \dom(\declinds)$ the interpretation of the
  arguments of the arity of $\declcons(\indtyp)$ are subsets of the
  interpretation of corresponding arguments of the arity of
  $\declcons(\indtyp')$. Similarly for the arguments of the
  constructors. In addition, assume that in each case, the
  interpretation of values given as parameters and arities in the
  resulting type of each corresponding constructors (the inductive
  type being constructed by that constructor) are equal. These are conditions in the rule
  \textsc{Ind-leq} in
  Figure~\ref{fig:pCuIC:cumu-jueq-inducrtive-types} where cumulativity
  (subtyping) relation is replaced with subset relation on the
  interpretation and the judgemental equality is replaced with
  equality of interpretations.\footnote{These conditions will hold by the
    induction hypothesis when we shall use this lemma.}

  Let $\vv{\TypP}$ and $\vv{\TypP'}$ be parameters of inductive blocks
  $\Inds$ and $\Inds'$, respectively. In this case,
  \begin{align*}
    & \forall \indtyp \in \dom(\declinds).\;
      \declinds(d) \synteq \Forall \vv{\termP} : \vv{\TypP}. \vv{\vartermC} : \vv{\Term}. \arity_{\indtyp}
      \Rightarrow \\
    & \forall \vv{\genel}.\; \vv{\genel} \in \semtyped{\ctx}{\vv{\TypP}}{\env}
      \land \vv{\genel} \in \semtyped{\ctx}{\vv{\TypP'}}{\env} \Rightarrow\\
    & \forall \vv{\genelB} \in \semtyped{\ctx, \vv{\termP} : \vv{\TypP}}{\vv{\Term}}{\env} \Rightarrow
      \vv{\decode}(\semtyped{\ctx}{\Inds.\indtyp}{\env}, \vv{\genel}, \vv{\genelB})
      \subseteq \vv{\decode}(\semtyped{\ctx}{\Inds'.\indtyp}{\env},
      \vv{\genel}, \vv{\genelB})
  \end{align*}
\end{lemma}

\begin{proof}
  Expanding the definition of $\semtyped{\ctx}{\Inds.\indtyp}{\env}$
  and $\semtyped{\ctx}{\Inds'.\indtyp}{\env}$ in the statement above
  gives us:
  \begin{align*}
    & \forall \indtyp \in \dom(\declinds).\;
      \declinds(d_i) \synteq \Forall \vv{\termP} : \vv{\TypP}. \vv{\vartermC} : \vv{\Term}. \arity_{\indtyp_i}
      \Rightarrow \\
    & \forall \vv{\genel}.\; \vv{\genel} \in \semtyped{\ctx}{\vv{\TypP}}{\env}
      \land \vv{\genel} \in \semtyped{\ctx}{\vv{\TypP'}}{\env} \Rightarrow\\
    & \forall \vv{\genelB} \in \semtyped{\ctx, \vv{\termP} : \vv{\TypP}}{\vv{\Term}}{\env} \Rightarrow
      \set{c \middle| \tuple{i; \vv{\genel}; \vv{\genelB};c} \in \semtyped{\ctx}{\Inds}{\env}}
      \subseteq \set{c \middle| \tuple{i; \vv{\genel}; \vv{\genelB};c} \in \semtyped{\ctx}{\Inds'}{\env}}
  \end{align*}
  In order to show this, we show, by transfinite induction on $\ord$ up to the closing ordinal of
  $\ruleset_{\ctx, \Inds}^{\env}$, that the following holds
  \begin{align*}
    & \forall \indtyp \in \dom(\declinds).\;
      \declinds(d_i) \synteq \Forall \vv{\termP} : \vv{\TypP}. \vv{\vartermC} : \vv{\Term}. \arity_{\indtyp_i}
      \Rightarrow \\
    & \forall \vv{\genel}.\; \vv{\genel} \in \semtyped{\ctx}{\vv{\TypP}}{\env}
      \land \vv{\genel} \in \semtyped{\ctx}{\vv{\TypP'}}{\env} \Rightarrow\\
    & \forall \vv{\genelB} \in \semtyped{\ctx, \vv{\termP} : \vv{\TypP}}{\vv{\Term}}{\env} \Rightarrow
      \set{c \middle| \tuple{i; \vv{\genel}; \vv{\genelB};c} \in \TFoperation{\ruleset_{\ctx, \Inds}^{\env}}{\ord}}
      \subseteq \set{c \middle| \tuple{i; \vv{\genel}; \vv{\genelB};c} \in \TFoperation{\ruleset_{\ctx, \Inds'}^{\env}}{\ord}}
  \end{align*}

  The base case, and the case for limit ordinals are trivial.  In case
  of a successor ordinal, $\ord^{+}$, let a $k$ be in
  $\set{c \middle| \tuple{i; \vv{\genel}; \vv{\genelB};c} \in
    \TFoperation{\ruleset_{\ctx, \Inds}^{\env}}{\ord^{+}}}$. Then it
  must be generated by a rule in $\ruleset_{\ctx, \Inds}^{\env}$. By
  the premise of this rule, we know that all the recursive (the same
  or other inductive types of the block) arguments are required to be
  in the previous stage,
  $\TFoperation{\ruleset_{\ctx, \Inds}^{\env}}{\ord}$. By induction
  hypothesis, we know that those tuples should also belong to
  $\TFoperation{\ruleset_{\ctx, \Inds'}^{\env}}{\ord}$. Hence, the
  corresponding rule exists in $\ruleset_{\ctx, \Inds'}^{\env}$ and is
  applicable. Therefore, the tuple $k$ must also be in
  $\set{c \middle| \tuple{i; \vv{\genel}; \vv{\genelB};c} \in
    \TFoperation{\ruleset_{\ctx, \Inds'}^{\env}}{\ord^{+}}}$.
\end{proof}

\subsection{Modeling eliminators}
We define the interpretation for eliminators by constructing a rule
set which will interpret the eliminators for the whole inductive
block. We will define interpretation of individual eliminators based
on the fixpoint of this rule set.

\begin{definition}[Interpretation of recursors]\label{def:interp_recursors}
  Let $\Inds \synteq \Ind_{n}\{\declinds := \declcons\}$ be an
  arbitrary but fixed inductive block and assume, without loss of
  generality, that
  $\declinds = \indtyp_0 : \Typ_0, \dots \indtyp_l : \Typ_l$ and
  $\declcons = \constr_0 : \varTyp_0, \dots \constr_{l'} :
  \varTyp_{l'}$.  Where
  $\Typ_i \synteq \Forall \vv{\termP} : \vv{\TypP}. \Forall \vv{\var}
  : \vv{\TypB}. \sort$ for some types $\vv{\TypB}$ and some sort
  $\sort$. Here, $\vv{\TypP}$ are parameters of the inductive
  block. The type of constructors are of the following form:
  $\varTyp_k \synteq \Forall \vv{\termP} : \vv{\TypP}. \Forall
  \vv{\var} : \vv{\varTypC_k}. \indtyp_{i_k}~\vv{\termP}~\vv{\term_k}$
  for some $\vv{\term_k}$. For the sake of simplicity of presentation,
  let us write
  $\mathit{ELB}\synteq\ElimBlock{\Inds}{\motive_{\indtyp_1}, \dots,
    \motive_{\indtyp_l}}{\genfun_{\constr_1},\dots,\genfun_{\constr_{l'}}}$
  for the block of eliminators being interpreted. Furthermore, let
  $\elimtyp{\Inds}{\motive_{\indtyp_1}, \dots,
    \motive_{\indtyp_l}}(\constr_k, \varTyp_k) \synteq \Forall
  \vv{\var} :\vv{\varTypB_k}. \motive_{\indtyp_{i_k}}~\vv{\termB}$ for
  some terms $\vv{\termB}$. Notice that
  $\elimtyp{\Inds}{\motive_{\indtyp_1}, \dots, \motive_{\indtyp_l}}$
  is simply the type of the constructor where after each (mutually)
  recursive argument, an argument is added for the result of the
  elimination of that (mutually) recursive argument. Let us write
  $J_{k, i}$ for the index of the $i$\textsuperscript{th} argument of
  the $k$\textsuperscript{th} constructor, $\varTypC_{k, i}$, in
  $\vv{\varTypB_{k}}$ above. More precisely, whenever
  $\recarg(\varTypC_{k, i})$ holds, we have $\varTypB_{J_{k, i}+1}$ is
  the argument of
  $\elimtyp{\Inds}{\motive_{\indtyp_1}, \dots,
    \motive_{\indtyp_l}}(\constr_k, \varTyp_k)$ that corresponds to
  the result of the elimination of
  $\varTypC_{k, i} = \varTypB_{J_{k, i}}$.  We first define a rule set
  $\ruleset_{\ctx, \mathit{ELB}}^{\env}$ for this interpretation:
  \begin{align*}
    \ruleset_{\ctx, \mathit{ELB}}^{\env} \eqdef{}
    & \bigcup_{\indtyp_i \in
      \dom(\declinds)}\bigcup_{\constr_k \in
      \constrsof(\indtyp_i)}
      \set{ \Rule{\Psi_{\indtyp_i,\constr_k}}{\psi_{\indtyp_i,\constr_k}}  \middle|
      \vv{\varterm} \in
      \semtyped{\ctx}{\vv{\varTypB_k}}{\env}
      }
      \intertext{Let $\vv{\vartermB}$ be the subsequence of $\vv{\varterm}$ corresponding to arguments of the constructor, i.e., it is obtained from $\vv{\varterm}$ by dropping any term in the sequence that corresponds to some $\varTypB_{J_{k, j}+1}$ whenever $\recarg(\varTypC_{{k, j}})$.}
      \psi_{\indtyp_i,\constr_k} \eqdef{}
    & \tuple{\tuple{\constr_k; \vv{\vartermB}};
      \vv{\decode}(\semtyped{\ctx}{\genfun_{\constr_k}}{\env}, \vv{\varterm})}\\
    \Psi_{\indtyp_i,\constr_k} \eqdef{}
    &
      \bigcup_{\recarg(\varTypC_{k,j})}\set{
      \tuple{\begin{array}{l}
               \vv{\decode}(\semtyped{\ctx_{k, j}}{\varTypB_{J_{k, j}}}{\env, \vv{\vartermD}}, \vv{\termB});\\
             \vv{\decode}(\semtyped{\ctx_{k, j}}{\varTypB_{J_{k, j}+1}}{\env, \vv{\vartermD}}, \vv{\termB})\end{array}}
      \middle|
      \vv{\termB} \in \semtyped{\ctx_{k,j}}{\vv{\varTypD_{k, j}}}{\env, \vv{\vartermD}}}\\
    \ctx_{k,j} \eqdef{} & \ctx, \vv{\termP} : \vv{\TypP}, \var_1 : \varTypC_{k, 1}, \dots, \var_{j-1} : \varTypC_{k, j-1} \\
    \vv{\vartermD} \eqdef{}& \vartermB_1, \dots \vartermB_{j-1}\\
    \intertext{For $\Psi_k$ we assume that $\varTypC_{k, j} \synteq \Forall \vv{\varB} :
    \vv{\varTypD_{k,j}}. \indtyp_{I_{k, j}}~\vv{\termQ}~\vv{\term}$.}
  \end{align*}
  We define the interpretation individual eliminators as follows:
  \[
    \semtyped{\ctx}{\Elim{\term}{\Inds.\indtyp_i;
        \vv{\termC}}{\motive_{\indtyp_1}, \dots,
        \motive_{\indtyp_l}}{\genfun_{\constr_1},\dots,\genfun_{\constr_{l'}}}}{\env}
    \eqdef{} \termB\] if
  $\semtyped{\ctx}{\term}{\env} = \tuple{k, \vv{\vartermC}}$, and
  there is a unique $\termB$ such that
  $\tuple{\tuple{k; \semtyped{\ctx}{\vv{\termC}}{\env}, \vv{\vartermC}}; \termB} \in
  \inddef{\ruleset_{\ctx, \mathit{ELB}}^{\env}} $.
\end{definition}

The following lemma shows that in
Definition~\ref{def:interp_recursors} above we have indeed captured and
interpreted all of the elements of all of the inductive types in the
block.

\begin{lemma} \label{lem:interp_eliminator_correct}
  Let $\Inds \synteq \Ind_{n}\{\declinds := \declcons\}$ be a block of
  inductive types with $\semtyped{\ctx}{\Inds}{\env}$ defined and let
  $\motive_{\indtyp_1}, \dots, \motive_{\indtyp_l}$ and
  $\genfun_{\constr_1},\dots,\genfun_{\constr_{l'}}$ be such that
  \[\semtyped{\ctx}{\motive_{\indtyp_i}}{\env} \in
    \semtyped{\ctx}{\Forall \vv{\var}:\vv{\Typ}. (\indtyp_i~\vv{\var})
      \to \sort'}{\env}\] and
  \[\semtyped{\ctx}{\genfun_{\constr_i}}{\env} \in
    \semtyped{\ctx}{\elimtyp{\Inds}{\vv{\motive}}(\constr_i,
      \declcons(\constr_i))}{\env}\] Also, assume that for the term
  $\term$ we have
  $\semtyped{\ctx}{\term}{\env} \in
  \semtyped{\ctx}{\Inds.\indtyp~\vv{\termB}~\vv{\termC}}{\env}$, where
  $\vv{\termB}$ correspond to parameters.  Then,
  \[\Defined{\semtyped{\ctx}{\Elim{\term}{\Inds.\indtyp; \vv{\termB}}{\motive_{\indtyp_1},
          \dots,
          \motive_{\indtyp_l}}{\genfun_{\constr_1},\dots,\genfun_{\constr_{l'}}}}{\env}}\]
  and
  \[\semtyped{\ctx}{\Elim{\term}{\Inds.\indtyp; \vv{\termB}}{\motive_{\indtyp_1}, \dots,
        \motive_{\indtyp_l}}{\genfun_{\constr_1},\dots,\genfun_{\constr_{l'}}}}{\env}
    \in \semtyped{\ctx}{\motive_{\indtyp_i}~\vv{\termB}~\vv{\termC}~\term}{\env} \]
\end{lemma}

\begin{proof}
  We show by transfinite induction up to the closing ordinal of
  $\ruleset_{\ctx, \Inds}^{\env}$ that for any $\ord$ and for any
  $\tuple{\tuple{k;\vv{\varterm},\vv{\vartermC}}; \vv{\term}} \in
  \set{\tuple{\tuple{k;\vv{\varterm},\vv{\vartermC}}; \vv{\term}}
    \middle| \tuple{i; \vv{\varterm}; \vv{\term};\tuple{k;
        \vv{\vartermC}}} \in \TFoperation{\ruleset_{\ctx,
        \Inds}^{\env}}{\ord}}$ (note that by
  Lemma~\ref{lem:inerp_inductive_indices_unique} $\vv{\term}$ is
  uniquely determined by $k$, $\vv{\varterm}$ and $\vv{\vartermC}$)
  there is a unique
  $\genelB \in
  \semtyped{\ctx}{\vv{\decode}(\motive_{\indtyp_i},\vv{\varterm},\vv{\term},\tuple{k;\vv{\vartermC}})}{\env}$
  such that
  $\tuple{\tuple{\constr;\vv{\varterm},\vv{\vartermC}}; \genelB} \in
  \TFoperation{\ruleset_{\ctx, \mathit{ELB}}^{\env}}{\ord}$ for
  \[\mathit{ELB} \synteq \ElimBlock{\Inds}{\motive_{\indtyp_1}, \dots,
      \motive_{\indtyp_l}}{\genfun_{\constr_1},\dots,\genfun_{\constr_{l'}}}\]
  For the base case, $\ord = 0$ this holds trivially. For the other
  cases, it suffices to notice that all the argument elements taken
  for the (mutually) recursive arguments of constructors, i.e.,
  corresponding to $\varTypB_{J_{k, j}}$ for
  $\recarg(\varTypC_{J_{k, j}})$ are \emph{uniquely} determined
  from the arguments of the constructor ($\vv{\vartermB}$ in the
  interpretation of recursors above) as is so restricted by the
  antecedents of each rule.

  Notice that the argument above also shows that
  $\ruleset_{\ctx, \mathit{ELB}}^{\env}$ and
  $\ruleset_{\ctx, \Inds}^{\env}$ have the same closing ordinal. This
  concludes the proof.
\end{proof}

\subsection{Proof of Soundness}
We show that the model that we have constructed throughout this
section is sound. That is, we show that for any typing context $\ctx$,
term $\term$ and type $\Typ$ such that $\typed{\ctx}{\term}{\Typ}$ we
have that for any environment $\env \in \sem{\ctx}$,
$\Defined{\semtyped{\ctx}{\term}{\env}}$,
$\Defined{\semtyped{\ctx}{\Typ}{\env}}$ and that
$\semtyped{\ctx}{\term}{\env} \in \semtyped{\ctx}{\Typ}{\env}$.  We
use this result to prove consistency of \pCuIC{}.

\begin{lemma} \label{lem:model-closedness} Let $\ctx$ be a typing
  context, $\env$ be an environment, and $\term$ be a term such that
  $\Defined{\semtyped{\ctx}{\term}{\env}}$. Then
  $\freevars(\term) \subseteq \dom(\ctx)$.
\end{lemma}
\begin{proof}
  We prove that if there is a variable $\var \in \freevars(\term)$ such that $\var \not\in \dom(\ctx)$ then $\lnot \Defined{\semtyped{\ctx}{\term}{\env}}$. This follows easily by induction on $\term$.
\end{proof}

\begin{lemma}[Weakening] \label{lem:model-weakening} Let $\ctx$ be a
  typing context, $\env$ be an environment, $\term$ and $\Typ$ be
  terms such that $\Defined{\sem{\ctx, \ctx'}}$,
  $\env \in \sem{\ctx}$, $\env, \env' \in \sem{\ctx, \ctx'}$ and
  $\Defined{\semtyped{\ctx, \ctx'}{\term}{\env, \env'}}$. Furthermore,
  let $\ctxB$ be a typing context and $\envB$ be an environment such
  that $(\dom(\ctx) \cup \dom(\ctx')) \cap \dom(\ctxB) = \emptyset$,
  and we have that variables in $\dom(\ctxB)$ do not appear in $\ctx'$
  freely. Furthermore, let us assume that
  $\Defined{\sem{\ctx, \ctxB, \ctx'}}$ and
  $\env, \envB, \env' \in \sem{\ctx, \ctxB, \ctx'}$. Then
  \begin{enumerate}
  \item \label{lem:model-weakening:case1}
    $\Defined{\semtyped{\ctx, \ctxB, \ctx'}{\term}{\env, \envB, \env'}}$
  \item \label{lem:model-weakening:case2}
    $\semtyped{\ctx, \ctxB, \ctx'}{\term}{\env, \envB, \env'} =
    \semtyped{\ctx, \ctx'}{\term}{\env, \env'}$
  \end{enumerate}
\end{lemma}
\begin{proof}
  We prove the result above by induction on $\term$. Most cases follow
  easily by induction hypotheses. Here, for demonstration purposes we
  show the case for variables and (dependent) function types.
  \begin{itemize}
  \item $\term = \var$: In this case, both
    Case~\ref{lem:model-weakening:case1} and
    Case~\ref{lem:model-weakening:case2} above follow by definition of
    the model.
  \item $\term = \Forall \var : \Typ. \TypB$: We know by induction
    hypothesis that
    \[\Defined{\semtyped{\ctx, \ctxB, \ctx'}{\Typ}{\env, \envB,
        \env'}}\]
    \[\semtyped{\ctx, \ctxB, \ctx'}{\Typ}{\env, \envB, \env'} =
      \semtyped{\ctx, \ctx'}{\Typ}{\env, \env'}\]
    Also by induction hypotheses we have that for all
    $\genel \in \semtyped{\ctx, \ctxB, \ctx'}{\Typ}{\env, \envB,
      \env'}$,
    \[\Defined{(\semtyped{\ctx, \ctxB, \ctx', \var : \Typ}{\TypB}{\env,
        \envB, \env', \genel})}\]
    \[\semtyped{\ctx, \ctxB, \ctx', \var : \Typ}{\TypB}{\env, \envB,
      \env', \genel} = \semtyped{\ctx, \ctx', \var :
      \Typ}{\TypB}{\env, \env', \genel}\] Given the above induction
  hypotheses, both Case~\ref{lem:model-weakening:case1} and
  Case~\ref{lem:model-weakening:case2} above follow by the definition
  of the model. Note that by Lemma~\ref{lem:model-closedness}, we know
  that variables in $\dom(\ctxB)$ do not appear freely in $\Typ$. This
  is why induction hypothesis applies to $\TypB$ above. \qedhere
  \end{itemize}
\end{proof}

\begin{lemma}[Substitutivity] \label{lem:model-subst}
  Let $\ctx$ be a typing context, $\env$ be an environment,
  $\termB$ and $\Typ$ be terms such that $\Defined{\sem{\ctx}}$,
  $\env \in \sem{\ctx}$ and $\Defined{\semtyped{\ctx}{\termB}{\env}}$. Then
  \begin{enumerate}
  \item \label{lem:model-subst:case1} If
    $\Defined{\sem{\ctx, \var : \Typ, \ctxB}}$ and
    $\env, \semtyped{\ctx}{\termB}{\env}, \envB \in \sem{\ctx, \var :
      \Typ, \ctxB}$ then
    $\env, \envB \in \sem{\ctx, \subst{\ctxB}{\var}{\termB}}$
  \item \label{lem:model-subst:case2} If
    $\Defined{\sem{\ctx, \var : \Typ, \ctxB}}$ and
    $\env, \semtyped{\ctx}{\termB}{\env}, \envB \in \sem{\ctx, \var :
      \Typ, \ctxB}$ and
    $\Defined{\semtyped{\ctx, \var : \Typ, \ctxB}{\term}{\env,
      \semtyped{\ctx}{\termB}{\env}, \envB}}$ then
    \begin{enumerate}
    \item \label{lem:model-subst:case2a}
      $\Defined{\semtyped{\ctx,
          \subst{\ctxB}{\var}{\termB}}{\subst{\term}{\var}{\termB}{}}{\env,
          \envB}}$
    \item \label{lem:model-subst:case2b}
      $\semtyped{\ctx,
        \subst{\ctxB}{\var}{\termB}}{\subst{\term}{\var}{\termB}{}}{\env,
        \envB} =\\ \semtyped{\ctx, \var : \Typ, \ctxB}{\term}{\env,
        \semtyped{\ctx}{\termB}{\env}, \envB} = \semtyped{\ctx, \var
        := \termB : \Typ, \ctxB}{\term}{\env,
        \semtyped{\ctx}{\termB}{\env}, \envB}$
    \end{enumerate}
  \end{enumerate}
\end{lemma}

\begin{proof}
  We prove this lemma by well-founded induction on
  $\sigma(\ctxB, \term) = \size(\ctxB) + \size(\term) -
  \sfrac{1}{2}$. In the following we reason as though we are
  conducting induction on the structure of terms and contexts.  Note
  that this is allowed because our measure $\sigma(\ctxB, \term)$ does
  decrease for the induction hypotheses pertaining to structural
  induction. This is in particular crucial in some sub-cases of the
  case where $\ctxB = \ctxB, \varB : \TypB$.  There, we are performing
  induction on $\term$ which enlarges the context $\ctxB$ (similar to
  case of (dependent) functions in Lemma~\ref{lem:model-weakening}).
  \begin{itemize}
  \item Case $\ctxB = \cdot$: then Case~\ref{lem:model-subst:case1}
    holds trivially. Case~\ref{lem:model-subst:case2}, follows by
    induction on $t$. The only non-trivial case (not immediately
    following from the induction hypothesis) is the case where $\term$
    is a variable. In this case, both \ref{lem:model-subst:case2a} and
    \ref{lem:model-subst:case2b} follow by definition of the model.

  \item Case $\ctxB = \ctxB', \varB : \TypB$ and correspondingly,
    $\envB = \envB', \genel$:
    \begin{enumerate}
    \item[\ref{lem:model-subst:case1}.] Follows by definition of the
      model and the induction hypothesis corresponding to
      Case~\ref{lem:model-subst:case2b}:
      $\semtyped{\ctx,
        \subst{\ctxB'}{\var}{\termB}}{\subst{\TypB}{\var}{\termB}}{\env,
        \envB'} = \semtyped{\ctx, \var : \Typ, \ctxB'}{\TypB}{\env,
        \semtyped{\ctx}{\termB}{\env}, \envB'}$
    \item[\ref{lem:model-subst:case2}.] We proceed by induction on
      $\term$. The only non-trivial case is when $\term$ is a
      variable, $\term = \varC$. Notice, when $\varC \neq \var$ then
      both \ref{lem:model-subst:case2a} and
      \ref{lem:model-subst:case2b} hold trivially by definition of the
      model. Otherwise, we have to show that
      $\Defined{\semtyped{\ctx, \subst{\ctxB'}{\var}{\termB}, \varB :
          \subst{\TypB}{\var}{\termB}}{\termB}{\env, \envB, \genel}}$
      and
      $\semtyped{\ctx, \subst{\ctxB}{\var}{\termB}, \varB :
        \subst{\TypB}{\var}{\termB}}{\termB}{\env, \envB, \genel} =$
      $\semtyped{\ctx, \var : \Typ, \ctxB, \varB :
        \subst{\TypB}{\var}{\termB}}{\var}{\env,
        \semtyped{\ctx}{\termB}{\env}, \envB. \genel} =$\\
      $\semtyped{\ctx, \var := \termB : \Typ, \ctxB, \varB :
        \subst{\TypB}{\var}{\termB}}{\var}{\env,
        \semtyped{\ctx}{\termB}{\env}, \envB, \genel}$. These follow
      by Lemma~\ref{lem:model-weakening} above and our assumption that
      $\Defined{\semtyped{\ctx}{\termB}{\env}}$.
    \end{enumerate}
  \item Case $\ctxB = \ctxB', \Ind_{n}\{\declinds := \declcons\}$: We
    prove this case by induction on $\term$. The only non-trivial case
    here is the case where
    $\term = \Ind_{n}\{\declinds := \declcons\}.\varC$. Notice that
    the interpretation of
    $\semtyped{\ctx, \subst{\ctxB'}{\var}{\termB}, \varB :
      \subst{\TypB}{\var}{\termB}}{\term}{\env, \envB, \genel}$ is
    defined based on the interpretation of terms of the form
    $\semtyped{\ctx, \subst{\ctxB'}{\var}{\termB}}{\vartermC}{\env,
      \envB}$ where $\vartermC$ appears in $\declinds$ or $\declcons$
    to which the induction hypothesis (that of induction on $\ctxB$)
    applies. \qedhere
  \end{itemize}
\end{proof}

\begin{theorem}[Soundness of the model]\label{lem:model-soundness}
  The model defined in this section is sound for our typing
  system. That is, the following statements hold:
  \begin{enumerate}
    \item If $\WFctx{\ctx}$ then $\Defined{\sem{\ctx}}$
    \item If $\typed{\ctx}{\term}{\Typ}$ then $\Defined{\sem{\ctx}}$
      and for any $\env \in \sem{\ctx}$ we have
      $\Defined{\semtyped{\ctx}{\term}{\env}}$,
      $\Defined{\semtyped{\ctx}{\Typ}{\env}}$ and
      $\semtyped{\ctx}{\term}{\env} \in \semtyped{\ctx}{\Typ}{\env}$
    \item If $\jueq{\ctx}{\term}{\term'}{\Typ}$ then
      $\Defined{\semtyped{\ctx}{\term}{\env}}$,
      $\Defined{\semtyped{\ctx}{\term'}{\env}}$,
      $\Defined{\semtyped{\ctx}{\Typ}{\env}}$ and
      $\semtyped{\ctx}{\term}{\env} = \semtyped{\ctx}{\term'}{\env}
      \in \semtyped{\ctx}{\Typ}{\env}$
    \item If $\subtyp{\ctx}{\Typ}{\TypB}$ then
      $\Defined{\semtyped{\ctx}{\Typ}{\env}}$,
      $\Defined{\semtyped{\ctx}{\TypB}{\env}}$ and
      $\semtyped{\ctx}{\Typ}{\env} \subseteq
      \semtyped{\ctx}{\TypB}{\env}$
  \end{enumerate}
\end{theorem}

\begin{proof}
  Note that judgements of the form $\WFctx{\ctx}$,
  $\typed{\ctx}{\term}{\Typ}$, $\jueq{\ctx}{\term}{\term'}{\Typ}$ and
  $\subtyp{\ctx}{\Typ}{\TypB}$ are defined mutually. We prove the
  theorem by mutual induction on the derivation of these judgements.
  \begin{itemize}
    \item Case~\textsc{WF-ctx-empty}: trivial by definition.
    \item Case~\textsc{WF-ctx-hyp}: trivial by definition and induction hypothesis.
    \item Case~\textsc{WF-ctx-def}: trivial by definition and induction hypothesis.
    \item Case~\textsc{Prop}: trivial by definition.
    \item Case~\textsc{Hierarchy}: trivial by definition.
    \item Case~\textsc{Var}: trivial by definition and induction hypotheses.
    \item Case~\textsc{Let}: by definition, induction hypothesis and
      Lemma~\ref{lem:model-subst}.
    \item Case~\textsc{Let-eq}: by definition, induction hypotheses
      and Lemma~\ref{lem:model-subst}.
    \item Case~\textsc{Prod}: by induction hypotheses we know that
      $\semtyped{\ctx}{\Typ}{\env} \in \semtyped{\ctx}{\sort_1}{\env}$
      and that
      $\semtyped{\ctx, \var : \Typ}{\TypB}{\env, \genel} \in
      \semtyped{\ctx}{\sort_2}{\env}$ for any
      $\genel \in \semtyped{\ctx}{\Typ}{\env}$. We also know
      $\ProdRs(\sort_1, \sort_2, \sort_3)$. We have to show that
      \begin{align}\label{lem:model-soundness-prop-concl}
        \semtyped{\ctx}{\Forall \var : \Typ. \TypB}{\env} \in \semtyped{\ctx}{\sort_3}{\env}
      \end{align}
      Since $\ProdRs(\sort_1, \sort_2, \sort_3)$ holds, we have that
      either $\sort_2 = \sort_3 = \Prop$ or $\sort_1 = \Type{i}$,
      $\sort_2 = \Type{j}$ and $\sort_3 = \Type{\Max{i, j}}$ or
      $\sort_1 = \Prop$, $\sort_2 = \Type{i}$ and $\sort_3 = \Type{i}$
      hold. In the first case, The membership relation above,
      (\ref{lem:model-soundness-prop-concl}), follows from
      Case~\ref{lem:impred-encoding:case1} of
      Lemma~\ref{lem:impred-encoding}. In the other two cases, note
      that $\semtyped{\ctx}{\sort_3}{\env}$ is a von Neumann universe
      and is hence closed under (dependent) function space and also
      (by axiom schema of replacement) under taking trace-encoding of
      elements of any set.
    \item Case~\textsc{Prod-eq}: by definition and induction hypotheses.
    \item Case~\textsc{Lam}: by definition and induction hypotheses.
    \item Case~\textsc{Lam-eq}: by definition and induction hypotheses.
    \item Case~\textsc{App}: by induction hypotheses we know that
      $\semtyped{\ctx}{\TermB}{\env} \in \semtyped{\ctx}{\Typ}{\env}$
      and
      $\semtyped{\ctx}{\Term}{\env} \in \semtyped{\ctx}{\Forall \var :
        \Typ. \TypB}{\env}$. From the latter we know that
      $\semtyped{\ctx}{\Term}{\env} = \encode(\genfun)$ for some
      $\genfun$ with domain
      $\dom(\genfun) = \semtyped{\ctx}{\Typ}{\env}$ such that
      $\genfun(\genel) = \semtyped{\ctx, \var : \Typ}{\TypB}{\env,
        \genel}$.
      Therefore,
      \begin{align*}
        \semtyped{\ctx}{\Term~\TermB}{\env}
        ={} & \decode(\encode(f), \semtyped{\ctx}{\TermB}{\env})\\
        ={} & \genfun(\semtyped{\ctx}{\TermB}{\env}) \in \semtyped{\ctx,
          \var : \Typ}{\TypB}{\env, \semtyped{\ctx}{\TermB}{\env}}
        = \semtyped{\ctx}{\subst{\TypB}{\var}{\TermB}}{\env}
      \end{align*}
      where the last equality is by Lemma~\ref{lem:model-subst}.
    \item Case~\textsc{App-eq}: by a reasoning similar to that of
      Case~\textsc{App} above.
    \item Case~\textsc{Ind-WF}: by definition, induction hypotheses
      and Lemma~\ref{lem:interp_inductive_correct}.
    \item Case~\textsc{Ind-type}: by induction hypotheses we know that
      $\Defined{\sem{\ctx}}$ which by definition means that for any
      environment $\env \in \sem{\ctx}$ and any inductive block
      $\Inds \in \ctx$ we have (using the weakening lemma above)
      $\Defined{\semtyped{\ctx}{\Inds}{\env}}$.  By
      Lemma~\ref{lem:interp_inductive_correct} the desired result
      follows.
    \item Case~\textsc{Ind-constr}: by induction hypotheses we know
      that $\Defined{\sem{\ctx}}$ which by definition means that for
      any environment $\env \in \sem{\ctx}$ and any inductive block
      $\Inds \in \ctx$ we have (using the weakening lemma above)
      $\Defined{\semtyped{\ctx}{\Inds}{\env}}$. The desired result
      follows by the definition of
      $\semtyped{\ctx}{\Inds.\constr}{\env}$,
      Lemma~\ref{lem:model-subst} and the fact that
      $\Defined{\semtyped{\ctx}{\Inds}{\env}}$ is the fixpoint of the
      rule-set used to construct it.
    \item Case~\textsc{Ind-Elim}: by
      Lemma~\ref{lem:interp_eliminator_correct} and induction
      hypotheses.
    \item Case~\textsc{Ind-Elim-eq}: similarly to Case~\textsc{Ind-Elim}.
    \item Case~\textsc{Eq-ref}: trivial by definition and induction hypothesis.
    \item Case~\textsc{Eq-sym}: trivial by definition and induction hypothesis.
    \item Case~\textsc{Eq-trans}: trivial by definition and induction hypothesis.
    \item Case~\textsc{Beta}: by induction hypotheses that
      \begin{itemize}
        \item[] $\Defined{\semtyped{\ctx, \var : \Typ}{\Term}{\env}}$
        \item[] $\forall \genel \in \semtyped{\ctx}{\Typ}{\env}.\; \semtyped{\ctx, \var : \Typ}{\Term}{\env, \genel} \in
          \semtyped{\ctx, \var : \Typ}{\TypB}{\env, \genel}$
        \end{itemize}
        On the other hand, we have that
        \begin{itemize}
          \item[] $\Defined{\semtyped{\ctx}{\TermB}{\env}}$
          \item[] $\semtyped{\ctx}{\TermB}{\env} \in \semtyped{\ctx}{\Typ}{\env}$
          \item[] $\Defined{\semtyped{\ctx}{\Typ}{\env}}$
        \end{itemize}
        and therefore
        \[ \semtyped{\ctx, \var : \Typ}{\TypB}{\env,
            \semtyped{\ctx}{\TermB}{\env}} \in
          \semtyped{\ctx}{\sort}{\env} \] Also, by
        Lemma~\ref{lem:model-subst} we get that
      \[ \semtyped{\ctx}{\subst{\TypB}{\var}{\TermB}}{\env} =
        \semtyped{\ctx, \var : \Typ}{\TypB}{\env,
          \semtyped{\ctx}{\TermB}{\env}} \in
        \semtyped{\ctx}{\sort}{\env}\] and that
      \[\semtyped{\ctx}{\subst{\Term}{\var}{\TermB}}{\env} = \semtyped{\ctx, \var : \Typ}{\Term}{\env, \semtyped{\ctx}{\termB}{\env}} \in \semtyped{\ctx,
          \var : \Typ}{\TypB}{\env, \semtyped{\ctx}{\termB}{\env}}\]
      We finish the proof by:
      \begin{align*}
        \semtyped{\ctx}{(\Lam \var : \Typ. \Term)~\TermB}{\env} ={}
        & \decode(\semtyped{\ctx}{\Lam \var : \Typ. \Term}{\env}, \semtyped{\ctx}{\TermB}{\env})\\
        ={} & \decode(\encode(\genel \in \semtyped{\ctx}{\Typ}{\env} \mapsto
              \semtyped{\ctx, \var : \Typ}{\Term}{\env, \genel}), \semtyped{\ctx}{\TermB}{\env})\\
        ={} & \semtyped{\ctx, \var : \Typ}{\Term}{\env, \semtyped{\ctx}{\TermB}{\env}}\\
        ={} & \semtyped{\ctx}{\subst{\Term}{\var}{\TermB}}{\env}
      \end{align*}
      
    \item Case~\textsc{Eta}: We need to show that
      \[\semtyped{\ctx}{\term}{\env} = \semtyped{\ctx}{\Lam \var :
          \Typ. \term~\var}{\env} \in \semtyped{\ctx}{\Forall \var :
          \Typ. \TypB}{\env}\]  We know by induction hypothesis that
      $\Defined{\semtyped{\ctx}{\term}{\env}}$ and that
      $\semtyped{\ctx}{\term}{\env} \in \semtyped{\ctx}{\Forall \var :
        \Typ. \TypB}{\env}$ and consequently we know that there is a
      set theoretic function $\genfun$:
      \[\genfun \in \DepFun \genel in
        \semtyped{\ctx}{\Typ}{\env}. \semtyped{\ctx, \var :
          \Typ}{\TypB}{\env, \genel}\] such that
      $\semtyped{\ctx}{\term}{\env} = \encode(\genfun)$. This implies
      that,
      \begin{align*}
        \semtyped{\ctx}{\Lam \var : \Typ. \term~\var}{\env}
        & = \encode\left(\set{(\genel, \semtyped{\ctx, \var : \Typ}{\term~\var}{\env, \genel}) \middle| \genel \in \semtyped{\ctx}{\Typ}{\env}}\right) \\
        & = \encode\left(\set{(\genel, \decode(\semtyped{\ctx, \var : \Typ}{\term}{\env, \genel}, \genel)) \middle| \genel \in \semtyped{\ctx}{\Typ}{\env}}\right) \\
        & = \encode\left(\set{(\genel, \decode(\semtyped{\ctx}{\term}{\env, \genel}, \genel)) \middle| \genel \in \semtyped{\ctx}{\Typ}{\env}}\right)\\
        & = \encode\left(\set{(\genel, \decode(\encode(\genfun), \genel)) \middle| \genel \in \semtyped{\ctx}{\Typ}{\env}}\right) \\
        & = \encode\left(\set{(\genel, \genfun(\genel)) \middle| \genel \in \semtyped{\ctx}{\Typ}{\env}}\right) \\
        & = \encode(\genfun)\\
        & = \semtyped{\ctx}{\term}{\env}
        \end{align*}
    \item Case~\textsc{Delta}: by Lemma~\ref{lem:model-subst} and induction hypotheses.
    \item Case~\textsc{Zeta}: trivial by definition and induction hypothesis.
    \item Case~\textsc{Iota}: by definition of the interpretation of
      recursors and induction hypothesis. Notice that by construction
      of the interpretation of eliminators in
      Definition~\ref{def:interp_recursors} the interpretation of
      elimination of $\constr_i~\vv{\varterm}$ is basically, the
      result of applying $\genfun_{\constr_i}~\vv{\vartermB}$ where
      $\vv{\varterm}$ is a subsequence of $\vv{\vartermB}$. The only
      difference between $\vv{\vartermB}$ and $\vv{\varterm}$ is that
      in $\vv{\vartermB}$ after each value that corresponds to a
      (mutually) recursive argument of the constructor $\constr_i$ we
      have a value that corresponds to the interpretation of elimination
      of that (mutually) recursive argument. For those terms,
      $\recor{\Inds}{\vv{\motive}}{\vv{\genfun}}$ applies the
      elimination using the eliminator, $\mathbfsf{Elim}$, while the
      rule in rule set corresponding $\constr_i~\vv{\varterm}$ ensures
      in its antecedents that all elements are eliminated correctly
      according to the interpretation of the eliminator (the fixpoint
      taken in Definition~\ref{def:interp_recursors}). Notice that if
      the constructor is of an inductive sub-type, then, by
      construction of interpretation of constructors the
      interpretation of the two constructors applied to those terms
      are \emph{equal}.
    \item Case~\textsc{Prop-in-Type}: trivial by definition.
    \item Case~\textsc{Cum-Type}: trivial by definition.
    \item Case~\textsc{Cum-trans}: trivial by definition and induction hypothesis.
    \item Case~\textsc{Cum-weaken}: trivial by definition, induction
      hypothesis and Lemma~\ref{lem:model-weakening}.
    \item Case~\textsc{Cum-Prod}: trivial by definition and induction hypothesis.
    \item Case~\textsc{Cum-eq-L}: trivial by definition and induction hypothesis.
    \item Case~\textsc{Cum-eq-R}: trivial by definition and induction hypothesis.
    \item Case~\textsc{Cum}: trivial by definition and induction hypothesis.
    \item Case~\textsc{Cum-eq}: trivial by definition and induction hypothesis.
    \item Case~\textsc{Eq-Cum}: trivial by definition and induction hypothesis.
    \item Case~\textsc{C-Ind}: easy by definition, induction
      hypothesis and Lemma~\ref{lem:c-ind-in-model}.
    \item Case~\textsc{Ind-Eq}: by definition, induction hypotheses
      and Lemma~\ref{lem:c-ind-in-model}.
    \item Case~\textsc{Constr-Eq-L}: by definition, induction
      hypothesis.
    \item Case~\textsc{Constr-Eq-R}: by definition, induction
      hypothesis.\qedhere
  \end{itemize}
\end{proof}

\begin{corollary}[Consistency of \pCuIC{}]
  Let $\sort$ be a sort, then, there \emph{does not exist} any term $\term$
  such that $\typed{\cdot}{\term}{\Forall \var : \sort. \var}$.
\end{corollary}
\begin{proof}
  If there where such a term $\term$ by
  Theorem~\ref{lem:model-soundness} we should have
  $\semtyped{\cdot}{\term}{\nil} \in \semtyped{\cdot}{\Forall \var :
    \sort. \var}{\nil}$. However, $\semtyped{\cdot}{\Forall \var :
    \sort. \var}{\nil} = \emptyset$.
\end{proof}

\paragraph{Strong normalization}
We believe that our extension to \pCIC{} maintains strong
normalization and that the model constructed by \citet{barras-habil}
for \pCIC{} could be easily extended to support our added rules.

\subsection{The model and axioms of type theory}
Although our system does not explicitly feature any of the axioms
mentioned below, they are consistent with the model that we have
constructed.

Our model is a proof-irrelevant model. That is, all provable
propositions (terms of type $\Prop$) are interpreted identically.
Therefore, it satisfies the axiom of proof irrelevance and also the
axiom of propositional extensionality (that any two logically
equivalent propositions are equal). This model also satisfies
definitional/judgemental proof irrelevance for proposition. This is
similar to how Agda treats irrelevant arguments \cite{Abel2011}.

We do not support inductive types in the sort $\Prop$ in our
system. However, if the Paulin-style equality is encoded using
inductive types in higher sorts, then the interpretation of these
types would simply be collections of reflexivity proofs. Hence, our
model supports the axiom UIP (unicity of identity proofs) and
consequently all other logically equivalent axioms, e.g., axiom K
\cite{streicher1993investigations}.

This model, being a set theoretic model, also supports the axiom of
functional extensionality as set theoretic functions are
extensional. This is indeed why our model supports $\eta$-equivalence.

All these axioms are also supported by the model constructed by
\citet{DBLP:journals/corr/abs-1111-0123}.

\paragraph{Acknowledgements}
This work was partially supported by the CoqHoTT ERC Grant
637339 and partially by the Flemish Research Fund grants G.0058.13
(until June 2017) and G.0962.17N (since July 2017).

\bibliographystyle{abbrvnat}
\bibliography{bib}

\end{document}